%% file: paper.tex
\documentclass[conference]{IEEEtran}
\usepackage{etoolbox}
\usepackage{url}
\usepackage{hyperref}
\hypersetup{
  colorlinks   = true, %Colours links instead of ugly boxes
  urlcolor     = black!50!red, %Colour for external hyperlinks
  linkcolor    = black!50!brown, %Colour of internal links
  citecolor   = black!50!brown %Colour of citations
}

\usepackage{amsmath}
\usepackage{amssymb}
\usepackage{amsthm}
\usepackage{verbatim}
\usepackage{bm}
\usepackage{color,graphicx,xcolor}
\usepackage{mdwtab}
\usepackage{subfigure}
\usepackage{mathtools,tikz}
\mathtoolsset{showonlyrefs} 
\usepackage{hhline}
\usepackage{multirow}
\usepackage{pdfpages}
\usepackage{enumitem}
\IEEEoverridecommandlockouts
\onecolumn

\usepackage{cite}
\usepackage{amsmath,amssymb,amsfonts,amsthm,xspace}
\usepackage{algorithmic}
\usepackage{graphicx}
\usepackage{textcomp}

\usetikzlibrary{arrows,shapes.geometric}
\allowdisplaybreaks
\newtheorem{remark}{Remark}
\newtheorem{exam}{Example}
\newtheorem{thm}{Theorem}

\newtheorem{lemma}{Lemma}
\newtheorem{corollary}{Corollary}
\newtheorem*{dec*}{Decoder}
\newcommand{\subparagraph}[1]{\par {\em\underline{#1:}}}

\usepackage{pgffor}
\foreach \x in {a,...,z}{%
\expandafter\xdef\csname vec\x \endcsname{\noexpand\ensuremath{\noexpand\bm{\x}}}
}

\foreach \x in {A,...,Z}{%
\expandafter\xdef\csname vec\x \endcsname{\noexpand\ensuremath{\noexpand\bm{\x}}}
}

\foreach \x in {A,...,Z}{%
\expandafter\xdef\csname c\x \endcsname{\noexpand\ensuremath{\noexpand\mathcal{\x}}}
}

\foreach \x in {A,...,Z}{%
\expandafter\xdef\csname bb\x \endcsname{\noexpand\ensuremath{\noexpand\mathbb{\x}}}
}

\newcommand{\cavc}{CAVC}
\newcommand{\cs}{compound-state}
\newcommand{\Cs}{Compound-state}
\newcommand{\avcstate}{AVC-state}

\usepackage{blindtext}

\newcommand{\ww}{\ensuremath{{\overline{\mathcal{W}}_1 \cap \overline{\mathcal{W}}_2 }}}
\newcommand{\wu}{\ensuremath{{\overline{\mathcal{W}}_1 \cup \overline{\mathcal{W}}_2 }}}
\newcommand{\wk}{\ensuremath{{\overline{\mathcal{W}}_k} }}
\newcommand{\wa}{\ensuremath{{\overline{\mathcal{W}}_1} }}
\newcommand{\wb}{\ensuremath{{\overline{\mathcal{W}}_2} }}

\newcommand{\pn}[2]{\cP_{#1}^{(#2)}}

\def\Cdetcomm{C^{\mathsf{d}}_{\mathsf{com}}}
\def\Crancomm{C^{\mathsf{r}}_{\mathsf{com}}}
\def\Cdetboth{C^{\mathsf{d}}_{\mathsf{and}}}
\def\Cranboth{C^{\mathsf{r}}_{\mathsf{and}}}
\def\Cdetauth{C^{\mathsf{d}}_{\mathsf{or}}}
\def\Cranauth{C^{\mathsf{r}}_{\mathsf{or}}}

\def\Pand{\Phi^{\mathsf{and}}}
\def\Por{\Phi^{\mathsf{or}}}

\def\pand{\phi^{\mathsf{and}}}
\def\por{\phi^{\mathsf{or}}}

\begin{document}

\title{Compound Arbitrarily Varying Channels} 
\begin{comment}
\thanks{N. Sangwan and V. Prabhakaran acknowledge support of the Department of Atomic Energy, 
Government of India, under project no. RTI4001. 
N. Sangwan's work was additionally supported by the Tata Consultancy Services (TCS) 
foundation through the TCS Research Scholar Program.
Work of B. K. Dey was supported in part by Bharti Centre for Communication in IIT Bombay.
V. Prabhakaran's work was also supported by the 
Science \& Engineering Research Board, India through project MTR/2020/000308.}
\end{comment}

\author{\IEEEauthorblockN{Syomantak Chaudhuri}
\IEEEauthorblockA{IIT Bombay, India}
\and
\IEEEauthorblockN{Neha Sangwan}%\textsuperscript{2}}
\IEEEauthorblockA{TIFR, India}
\and
\IEEEauthorblockN{Mayank Bakshi}
\IEEEauthorblockA{Huawei, Hong Kong}
\and
\IEEEauthorblockN{Bikash Kumar Dey}
\IEEEauthorblockA{IIT Bombay, India}
\and
\IEEEauthorblockN{Vinod M. Prabhakaran}
\IEEEauthorblockA{TIFR, India}
}
\maketitle

\input{abstract}

\input{intro}
\section{System model} \label{sec:2}
\input{section2}
\section{Main results} \label{sec:results}
\input{results.tex}

\subsection{Communication over \cavc} \label{subs:comm}
\input{model_comm}
\subsection{Joint Communication and \Cs\ Identification over \cavc} \label{subs:both}
\input{model_both}
\subsection{Communication or \Cs\ Identification over \cavc} \label{subs:auth}
\input{model_auth}

\input{section3}

\input{proof_rand}
\input{proof_det}

\input{section4}
\input{ran_conv}
\input{ran_ach}
\input{det_ach}
\input{det_auth}

\section*{Acknowledgments}
N. Sangwan and V. Prabhakaran acknowledge support of the Department of Atomic Energy, 
Government of India, under project no. RTI4001. 
N. Sangwan's work was additionally supported by the Tata Consultancy Services (TCS) 
foundation through the TCS Research Scholar Program.
Work of B. K. Dey was supported in part by Bharti Centre for Communication in IIT Bombay.
V. Prabhakaran's work was also supported by the 
Science \& Engineering Research Board, India through project MTR/2020/000308.

\bibliography{refs} 
\bibliographystyle{ieeetr}

\appendix
\input{lemma-proof}

\end{document}

%% file: abstract.tex
\begin{abstract}
We propose a communication model, that we call compound arbitrarily varying channels
(\cavc), which unifies and generalizes compound channels and arbitrarily varying channels (AVC).
A \cavc\ can be viewed as a noisy channel with a fixed, but unknown,
\cs\ and an \avcstate\ which may vary with every channel use. 
The \avcstate\ is controlled by an adversary who is aware of the \cs. 
We study three problems in this setting: `communication', `communication
and \cs\ identification', and `communication or \cs\ identification'.
For these problems, we study conditions for feasibility and
capacity under deterministic coding and random coding.

%One of the most immediate question which arises in this setting is the new problem of adversary identification - the task of identifying the adversary which was active during a given transmission. 
%We combine the task of adversary identification with the task of communication and present the capacity expression for three distinct problems under both deterministic and random coding.
\end{abstract}

%\begin{IEEEkeywords}
%compound arbitrarily varying channel (\cavc ), 
%\cs\ identification, 
%\red{adversary identification},
%multiple-adversaries, 
%compound channel, 
%arbitrarily varying channel
%\end{IEEEkeywords}

%% file: intro.tex
%!TeX root=paper.tex
\section{Introduction} 

In communication systems modeled as discrete memoryless channels (DMC), 
it is assumed that the channel characteristics is fixed and known beforehand.
However, the compound DMC introduced by Blackwell et al. \cite{blackwell-compound}
models channels with fixed but unknown characteristics due to an unknown natural
state. Backwell et al. \cite{blackwell-AVC} also introduced arbitrarily varying channels (AVC)
where the channel state may vary arbitrarily in a
worst case manner for each symbol of transmission. The worst case
variation of the channel state in an AVC may be viewed as the act of a 
malicious adversary.

The capacity of a compound DMC was characterized in \cite{wolf}.
For AVC, the communication capacity under random coding was obtained
in \cite{blackwell-AVC}. The deterministic coding
capacity of an AVC is zero if the channel satisfies a
condition called {\it symmetrizability} which allows the 
adversary to mount an attack with a spurious message so as 
to confuse the decoder between this message and the sent message. 
When the channel is not symmetrizable, the deterministic coding 
capacity is the same as the random coding
capacity \cite{avc-csiszar-narayan}.

In this work, we consider a generalization where there is an unknown
\cs\ as well as an \avcstate\ determined by an adversary (see Figure~\ref{fig:CAVC}). 
The \cs\ is fixed over a blocklength of transmission, 
whereas the \avcstate\ may change for every symbol of transmission. 
%While desining the AVC state,
We assume that the adversary knows the \cs.
Associated with each \cs, the adversary has a set of channels that 
can be be instantiated (by setting the \avcstate). 
%Under each \cs\, the adversary can instantiate a channel from a family of channels.
We call this the Compound Arbitrarily Varying Channel (CAVC). 
This is a generalization of both compound channels and AVCs.
For simplicity, in this paper we only consider the case 
of two \cs s. 

We characterize the capacity of \cavc s under both random coding and
deterministic coding. For non-zero rates to be achievable under 
deterministic coding, first, the AVC under each \cs\ should be 
non-symmetrizable.
In addition, the channel should not satisfy a new condition, 
called {\em trans-symmetrizability}, which provides the adversary 
with an attack strategy that can confuse the decoder between the 
sent message under one \cs\ with another message under the other \cs\
(see Fig.~\ref{fig:trans}). We show that when a \cavc\ is not 
symmetrizable in either of these senses, the deterministic coding 
capacity is same as the random coding capacity.

\input{figure}

Another way to view the \cavc\ model is to associate an adversary 
with each \cs\ and exactly one of them being active 
for the entirety of the transmission. 
Associated with each adversary, there is a family of channels from
which it can instantiate a channel for each channel use.
%Each adversary can instantiate a channel from a family of channels. 
In such a situation,
it is also of interest to identify\footnote{Note that this is
significantly different from identifying an \textit{internal} 
adversary in a multiuser channel with byzantine
users \cite{NehaISIT21}.} 
the active adversary.
Thus, in addition to the communication problem, we also study two other
problems in the \cavc\ setup --
joint `communication {\em and} \cs\ identification' and
{`communication {\em or} \cs\ identification'}.
In the first (resp. second) problem above, the decoder needs to decode
the message and (resp. or) identify the \cs . 
In both these settings, we characterize the condition for non-zero rates
under deterministic codes
and also the capacities under deterministic coding and random coding.

If the \cs\ was known to the decoder, the CAVC model would be
a special case of arbitrarily varying broadcast 
channels~\cite{Jahn1981,PeregSteinberg2017,HofBross2006}. 
The trans-symmetrizability condition for
non-zero rates in a CAVC arises precisely because the decoder does not
know the \cs.
In \cite{kosut-kliewer,BKKGYu,GYuS}, on {authentication} 
in channels which may be controlled by an adversary, a relaxed 
decoding requirement is considered. 
When there is no adversary, the decoded message must be correct; but when
the adversary is active, the decoder is allowed to declare 
the presence of the adversary without decoding the message
(however, if the decoder outputs a message instead, it must be correct).
%The adversary in \cite{GYuS} has knowledge of the sent message as well as the code whereas the adversary in \cite{kosut-kliewer} and \cite{BKKGYu} is only aware of the code.
These models are close to our `communication {\em or} identification' model. In fact, we recover the result in 
\cite{kosut-kliewer} as a special case (see Remark~\ref{rem:KK}).
% when the adversary's alphabet under the first \cs\ is the no-adversary state and it is contained in the adversary's alphabet under second \cs. 
The work in \cite{KK2} considers communication in a Compound-Arbitrarily-Varying network 
where the adversary selects a subset of edges from a network which are then attacked with arbitrary transmissions. 
% which is significantly different from our present formulation. 
%A Gaussian two-hop network with a byzantine relay was considered in \cite{Yener} with the requirement of message decoding as well as byzantine attack detection. %The decoding guarantee in this model is closer to our second model\footnote{this is temporary. we need to see how to refer to these models.} where we require  message decoding as well as adversary identification. 

% Kosut and Kliewer studied an authenticated communication model in
% \cite{kosut-kliewer}, where the channel is either benign with no adversary,
% or has an adversary who picks a state in every symbol. The benigh channel
% corresponds to a no-adversary state. This is a special case of the
% communication or \cs\ identification problem in our \cavc\ model
% (see Remark~\ref{}).

In Section~\ref{sec:2}, we formally describe the \cavc\ model and present
the problems studied in this paper. We
present our results on the three problems in Sections \ref{subs:comm},
\ref{subs:both}, and \ref{subs:auth}. Section \ref{sec:proofs} 
provides proof sketches for the results.

%% file: figure.tex
\begin{figure}
    \centering
    \includegraphics[width=0.4\textwidth]{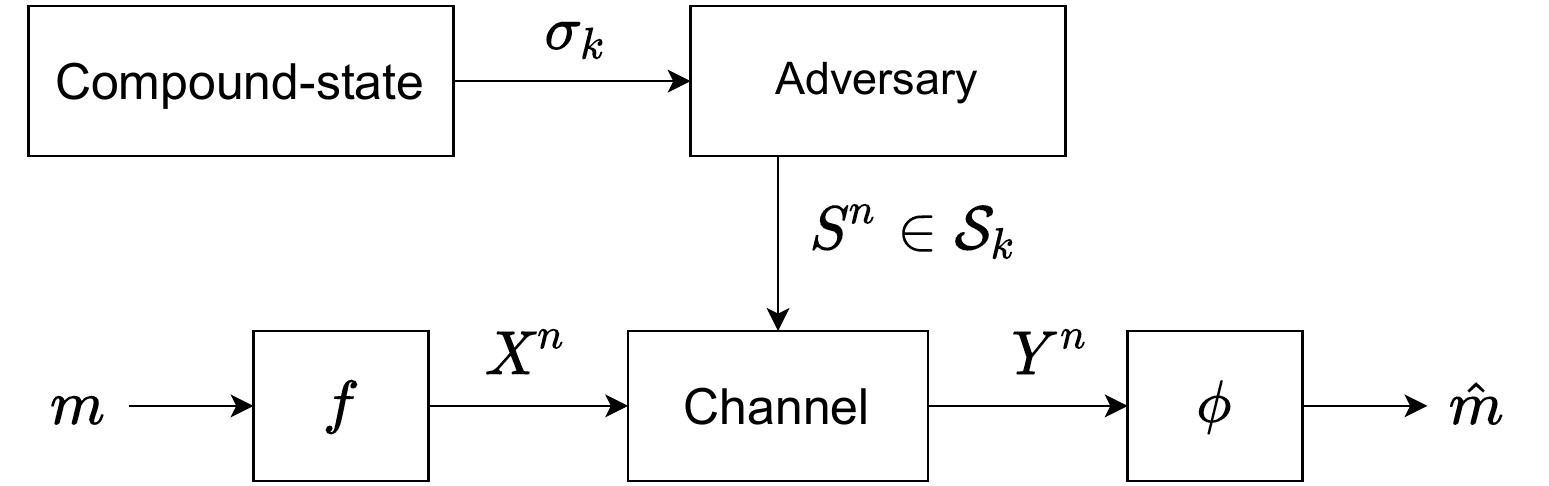}
    \caption{Compound Arbitrarily Varying Channel: The adversary knows the \cs\ $\sigma_k$
    and for each \cs, the adversary has a set of \avcstate s $\cS_k$. 
    The \cavc\ is modeled to be discrete memoryless and 
    the \cs\ remains fixed through out the transmission of a block.}
    \label{fig:CAVC}
\end{figure}

%% file: section2.tex
%\section{\textbf{Problem Setting and Main Results}} \label{sec:2}
%\setlength{\textfloatsep}{4pt plus 1.0pt minus 2.0pt}

\begin{table*}[!ht]
\centering
\caption{A brief summary of the problems studied and the results presented in this work.}
\label{tab:t1}
\begin{tabular}{|c|c|c|c|c|}
\hline
\textbf{Task} & \textbf{\begin{tabular}[c]{@{}c@{}} Output set\\ $\widehat \cM$ \end{tabular}} & \textbf{\begin{tabular}[c]{@{}c@{}} Error set \\ $\cE_{m,k}$ \end{tabular}} & \textbf{\begin{tabular}[c]{@{}c@{}} Conditions for positive \\ deterministic capacity \end{tabular}} & \textbf{\begin{tabular}[c]{@{}c@{}} Capacity \\ expression \end{tabular}} \\ \hline
Communication & $\cM$ & $\{m' \in \widehat \cM: m' \neq m \}$ & Non-any-sym. & $\max_{P_X}\min_{W \in \wu} I(X;Y)$ \\ \hline
\begin{tabular}[c]{@{}c@{}}Communication\\ and\\ \Cs\ Identification\end{tabular} & $\cM \times \{\sigma_1,\sigma_2\}$ & $\{m' \in \widehat \cM: m' \neq (m,\sigma_i) \}$ & \begin{tabular}[c]{@{}c@{}}Non-any-sym.\\ $\ww = \emptyset$ \end{tabular}  & $\max_{P_X}\min_{W \in \wu} I(X;Y)$ \\ \hline
\begin{tabular}[c]{@{}c@{}}Communication\\ or\\ \Cs\ Identification\end{tabular} & $\cM \cup \{\sigma_1,\sigma_2\}$ & $\{m' \in \widehat \cM: m' \notin \{m,\sigma_i\} \}$ & Non-trans-sym. & $\max_{P_X}\min_{W \in \ww} I(X;Y)$ \\ \hline
\end{tabular}
\end{table*}

\textit{Notation:} We use bold symbols like $\vecx$,$\vecy$ to denote vectors and capital letters like $X$,$Y$ to denote random variables with $P_X,P_Y$ denoting their distributions respectively. The $i$-th element of a vector $\vecy$ is denoted as $y_i$. 
For a vector $\vecx$, the notation $P_{\vecx}$ refers to its empirical
distribution.  
%The $\epsilon$-typical set is denoted by $\tau_X^{\epsilon} =
%\{\vecx:|P_{\vecx}(x)-P_X(x)|\leq \epsilon\  \forall x \in \cX\}$. In
%particular, $\tau_X$ denotes the typical set when $\epsilon=0$.
For any subset ${\cal B}$ in a finite dimensional space $\mathbb{R}^k$,
its convex closure is denoted by ${\bar{\cal B}}$.

%\textcolor{blue}{We use the notation $U:\cX \to \cS$ to denote a DMC $U(s|x)$
%with input $x \in \cX$ and output $s \in \cS$. \textbf{CLOSURE DEFN NEEDS
%CHANGE - DONT KNOW.} We denote the convex closure of the family of channels
%corresponding to each adversary $A_i$ ($i=1,2$) as
%\begin{multline}
%    \wi = \{Z(y|x): \exists P \in \cP_i, \\ Z(y|x) = \sum_{s \in \cS_i}P(s)W(y|x,s) \ \forall x,y \}.
%\end{multline}
%}

A discrete-memoryless Compound Arbitrarily Varying Channel (\cavc) with a 
finite input alphabet $\cX$, a finite output alphabet $\cY$, and two \cs s
$\sigma_1$ and $\sigma_2$ is described by two
families, $\cW_1$ and $\cW_2$, of channels with input alphabet $\cX$ and output
alphabet $\cY$. These families of channels correspond to the
\cs s $\sigma_1$ and $\sigma_2$ respectively. In each family, the channels are indexed
by a finite set
$\cS_k$ ($k=1,2$) called the \avcstate\ alphabet and, in particular, $\cW_k$
($k=1,2$) is a set of channels $\{W(\cdot|\cdot,s),\ s \in \cS_k\}$.  On input
$\vecx \in \cX^n$ over $n$ uses of the channel, $n \in
\{1,2,\ldots\}$, the probability of receiving  $\vecy \in \cY^n$ is given by
$W^n(\vecy|\vecx,\vecs) = \prod_{i=1}^n W(y_i|x_i,s_i)$ for some $\vecs \in
\cS_1^n \cup \cS_2^n$. 

%\red{Equivalently, 
%the channel can be viewed as being controlled by an adversary,
%who chooses an \avcstate\ from $\cS_i$, for some \cs\ $a_i$, $i=1,2$. 
%The \cs\ $a_i$ is fixed throughtout the block length, the adversary is aware 
%of the \cs\footnote{\red{Or, based on the heoretical guarentees that we desire,
% it can be viewed that the adversary is controlling the \cs\ as well.}}, and it can vary
% the \avcstate\ during the transmission adversarially. Note that 
% when $\cW_1 = \cW_2$, the \cavc\ is equivalent to a standard AVC.} 
% The adversary knows %the \cs being
%controlled by two adversaries -- $A_1$ with state alphabets $\cS_1$ and $A_2$
%with state alphabet $\cS_2$ -- such that exactly one adversary is active
%throughout the transmission. \textcolor{blue}{The active adversary can choose
%any $n$ length state vector from its alphabet (see Figure~\ref{fig:CAVC})}.

We study the \cavc\ under three distinct but closely-related problem 
settings as specified at the end of this section. 
In all three problems, the \cavc\ is analyzed under both deterministic 
and random (shared-randomness between encoder and decoder unknown to the 
adversary) coding regimes. An $(M,n)$ deterministic code is characterized by
\begin{enumerate}
    \item a message set $\cM = \{1,\dots,M\}$,
    \item an encoder $f:\cM \to \cX^n$, and
    \item a decoder $\phi: \cY^n \to \widehat\cM$.
\end{enumerate}
The set $\widehat \cM$ is different for the three problems, and 
is described later in this section. Table~\ref{tab:t1} gives a short description of each problem and the results we present. 
The problems are studied under the average probability of 
error and it is assumed that the adversary is unaware of the message sent by the 
transmitter but is aware of the encoder and decoder pair $(f,\phi)$ used for transmission.

Let $\cE_{m,k} \subseteq \widehat \cM $ correspond to the set of erroneous
outputs from the decoder when message $m$ is sent and $\sigma_k$ is the \cs .
$\cE_{m,k}$ depends on the problem definition and we specify it 
at the end of this section for each problem.
For $k=1,2$, define 
\begin{align}
    P_{\mathsf{e}}^{\mathsf{d}}(f,\phi,k) &= \max_{\vecs \in \cS_k^n} \frac{1}{M}\sum_{m=1}^M W^n(\phi^{-1}(\cE_{m,k}) |f(m),\vecs). \label{eq:fpi}  
\end{align}
The average probability of error $P_{\mathsf{e}}^{\mathsf{d}}(f,\phi)$ is given by 
\begin{align}
    P_{\mathsf{e}}^{\mathsf{d}}(f,\phi) &= \max\{ P_{\mathsf{e}}^{\mathsf{d}}(f,\phi,1), P_{\mathsf{e}}^{\mathsf{d}}(f,\phi,2)  \}. \label{eq:fp} 
\end{align} 
A rate $R$ is defined to be achievable under deterministic coding if there exists a sequence of 
$(2^{nR},n)$ deterministic codes $\{f^{(n)},\phi^{(n)} \}_{n=1}^{\infty}$ such that 
$P_{\mathsf{e}}^{\mathsf{d}}(f^{(n)},\phi^{(n)}) \to 0$ as $n \to \infty$. 
The \textit{deterministic code capacity} is defined as the supremum of all achievable rates 
under deterministic coding.

Let $\cF$ be the set of all encoders $f:\cM \to \cX^n$ and $\cG$ be the set of all 
decoders $\phi:\cY^n \to \widehat \cM$. An $(M,n)$ random code is given by the pair 
$(F,\Phi) \sim Q(f,\phi)$ where $Q$ is a distribution on $\cF \times \cG$.
The adversary has the knowledge of the distribution $Q$ but does not know the realisation 
of $(F,\Phi)$ used during the transmission and it is unaware of the transmitted 
message as well. For $k=1,2$, define
\begin{align}
    P_{\mathsf{e}}^{\mathsf{r}}(Q,k) &= \max_{\vecs \in \cS_k^n} \hspace{-2pt} \sum_{(f,\phi) \in \cF\times\cG} \hspace{-10pt}
    Q(f,\phi) \frac{1}{M} \sum_{m=1}^M W^n(\phi^{-1}(\cE_{m,k}) |f(m),\vecs). \label{eq:r-qi}
\end{align}
The average probability of error $P_{\mathsf{e}}^{\mathsf{r}}(Q)$ for a random code is given by 
\begin{align}
    P_{\mathsf{e}}^{\mathsf{r}}(Q) &= \max\{ P_{\mathsf{e}}^{\mathsf{r}}(Q,1), P_{\mathsf{e}}^{\mathsf{r}}(Q,2). \label{eq:r-q} \} 
\end{align} 
 
A rate $R$ is defined to be achievable under random coding if there exists a sequence of 
$(2^{nR},n)$ random codes $\{Q^{(n)} \}_{n=1}^{\infty}$ such that $P_{\mathsf{e}}^{\mathsf{r}}(Q^{(n)}) \to 0$ 
as $n \to \infty$. The \textit{random code capacity} is defined as the supremum of all achievable 
rates under random coding.

We now define the three specific problems.

\paragraph*{Communication over \cavc} In this problem, the decoder needs to 
reconstruct the encoded message. Therefore, the decoder's reconstruction
alphabet is $\hat\cM = \cM$ and the set $\cE_{m,k}$ of erroneous 
decoder outputs is given by
$$\cE_{m,k} = \{m' \in \widehat \cM: m' \neq m \}.$$

\paragraph*{Joint Communication and \Cs\ Identification over \cavc}
Here the decoder needs to reconstruct the encoded message, and also identify
the \cs . Hence
$\widehat\cM = \cM \times \{\sigma_1,\sigma_2\}$ and the set $\cE_{m,k}$ is given by
\[
\cE_{m,k} = \{m' \in \widehat \cM: m' \neq (m,\sigma_k) \}. \]

\paragraph*{Communication or \Cs\ Identification over \cavc}
Here the decoder needs to either reconstruct the encoded message or
identify the \cs .
Hence $\widehat \cM = \cM \cup \{\sigma_1,\sigma_2\}$ and the set $\cE_{m,k}$ is given by
\[ \cE_{m,k} = \{m' \in \widehat \cM: m' \notin \{m,\sigma_k\} \}.  \]

%% file: results.tex
We now present the main results on the three problems in three respective
subsections.

%% file: model_comm.tex
%\subsection{Communication over \cavc} \label{subs:comm}
We denote the \cavc\ capacity for the communication problem under 
deterministic coding as $\Cdetcomm$ and that under randomized coding 
as $\Crancomm$. 

Communication over a \cavc\ is closely related to communication over 
an Arbitrarily Varying Channel (AVC). An AVC from $\cX$ to $\cY$ is given by a 
set of channels $\{W(\cdot|\cdot,s),s\in\cS\}$ parameterized by the state 
alphabet $\cS$. The \avcstate\ of the channel can change arbitrarily 
during the transmission. A \cavc\ is an AVC when $\cS_1 = \cS_2$. 
Csiszar and Narayan in \cite{avc-csiszar-narayan} defined the notion of a 
\textit{symmetrizable} AVC and showed that the deterministic coding capacity of an AVC, $C^{\mathsf{d}}_{\mathsf{AVC}}$, is positive if and only 
if the channel is not symmetrizable. 
An AVC is symmetrizable if there exists some channel $U:\cX \to \cS$ such that
$\forall x,x' \in \cX, y \in \cY$,
\begin{equation} \label{eq:cissym-csiszar}
\sum_{s} U(s|x')W(y|x,s) =  \sum_{s} U(s|x)W(y|x',s).
\end{equation}
\textbf{Cis-symmetrizability:} 
For a \cavc, symmetrizability can be defined 
under each \cs. For $k=1$ or $2$, we define a \cavc\ 
to be \textit{$\cS_k$-symmetrizable} if there exists a channel 
$U: \cX \to \cS_k$ such that \eqref{eq:cissym-csiszar} 
holds $\forall x,x' \in \cX, y \in \cY$ (see Figure~\ref{fig:cis}). 
If the \cavc\ is $\cS_k$-symmetrizable for $k=1$ or $k=2$ or both, 
then we say the \cavc\ is \textit{cis-symmetrizable}.

\begin{figure}
 \center
  \includegraphics[width=0.39\textwidth]{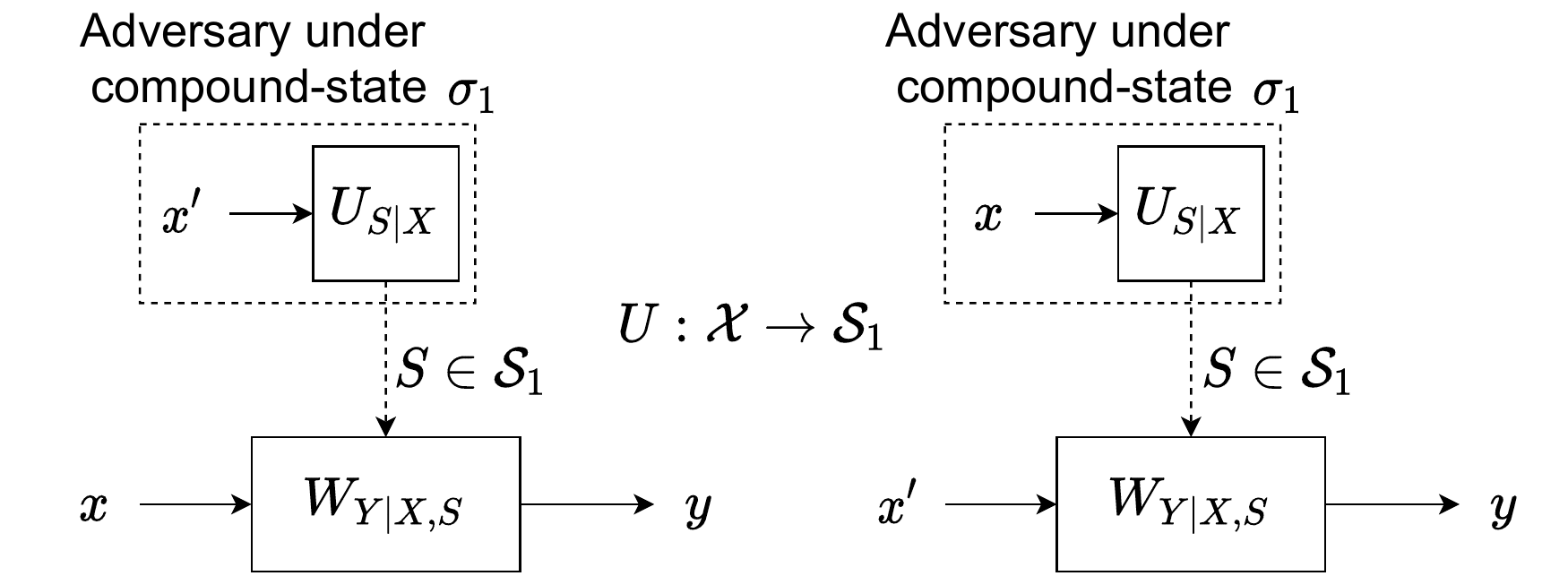}
  \caption{\textbf{$\cS_1$-symmetrizability:} 
  If there exists a channel $U:\cS_1 \to \cX$ such that the output distributions
  in the above two scenarios are the same for every pair of symbols $(x,x')\in \cX^2$
  then we call the channel $\cS_1$-symmetrizable.}
  \label{fig:cis}
\end{figure}

If the channel is $\cS_k$-symmetrizable and the \cs\ is $\sigma_k$, then for two 
distinct codewords $\vecx_m$, $\vecx_{m'}$ and $U$ satisfying \eqref{eq:cissym-csiszar}, 
the following two situations are indistinguishable :  
(i) the sender sends $\vecx_m$ and the adversary attacks 
when the \cs\ is $\sigma_k$ with \avcstate\ sequence from the output of the distribution 
$U^n(\cdot|\vecx_{m'})$ and (ii) the sender sends $\vecx_{m'}$ and the adversary 
attacks when the \cs\ is $\sigma_k$ with the output of the distribution $U^n(\cdot|\vecx_m)$. 
Thus, this argument is formalized in Section~\ref{sec:4} and it is possible to 
show that reliable decoding is not possible if a \cavc\ is cis-symmetrizable. \\
\textbf{Trans-symmetrizability:} The presence of two \cs s in a \cavc\ 
introduces another sufficient condition for $\Cdetcomm=0$ which we call 
trans-symmetrizability (see Figure~\ref{fig:trans}). 
Define a \cavc\ to be \textit{trans-symmetrizable} if 
there exists a pair of channels $U  : \cX \to \cS_1$, 
$V : \cX \to \cS_2$ such that $\forall x,x' \in \cX, y \in \cY$,
\begin{equation} \label{eq:transsym}
\sum_{s} U(s|x')W(y|x,s) =  \sum_{s} V(s|x)W(y|x',s).
\end{equation}
In a trans-symmetrizable \cavc\ with $U,V$ satisfying \eqref{eq:transsym} 
and $\vecx_m$, $\vecx_{m'}$ being distinct codewords, the following 
two situations are indistinguishable: (i) the sender sends codeword $\vecx_m$ 
and the adversary attacks when the \cs\ is $\sigma_1$ with the \avcstate\ sequence 
as the output of the distribution $U^n(\cdot|\vecx_{m'})$  and (ii) the sender 
sends codeword $\vecx_{m'}$ and the adversary attacks when the \cs\ is 
$\sigma_2$ with the state sequence as the output of the distribution 
$V^n(\cdot|\vecx_{m})$. Note that neither of cis- and trans-symmetrizability 
imply the other as demonstrated by the following two examples.

Consider a \cavc\ where $\cW_1$ with output alphabet $\cY_1$ and $\cW_2$ 
with output alphabet $\cY_2$ are symmetrizable AVCs satisfying 
$\cY_1 \cap \cY_2 = \emptyset$. Clearly, the \cavc\ is cis-symmetrizable 
but not trans-symmetrizable. Example~\ref{example:1} below presents a 
\cavc\ which is trans-symmetrizable, but not cis-symmetrizable.

\begin{exam}
\label{example:1}
Consider a \cavc\ with input alphabet $\cX$ and output alphabet $\cY$. Let $\cS_k = \cX \times \{k\}$. For $x \in \cX$ and $(x',k) \in \cS_k$, 
\begin{equation}
    y =   \begin{cases} (x,x') & \text{if } k=1, \\
    (x',x) & \text{if } k=2.
    \end{cases}
\end{equation}
This \cavc\ is clearly trans-symmetrizable using $U(s|x') = 1 $ 
if $s=(x',1)$ and $V(s|x) = 1$ if $s=(x,2)$. To show non-cis-symmetrizability, 
consider the case when the \cs\ is $\sigma_1$ and the input symbol is $x$. 
Since the channel reveals the input and the \avcstate\ completely when  
the \cs\ is $\sigma_k$, $k=1,2$, it cannot be cis-symmetrizable.
\end{exam}

We call a \cavc\ \textit{any-symmetrizable} if it is cis-symmetrizable or 
trans-symmetrizable (or both). Note that if a \cavc\ is any-symmetrizable then 
$\Cdetcomm = 0$. 
Further, for a \cavc\ with $\cW_k$ being the family of channels corresponding 
to \cs\ $\sigma_k$ , the capacity of the AVC with the family of channels 
$\cW = \cW_1 \cup \cW_2$ given by 
$$C^{\mathsf{d}}_{\mathsf{AVC}} = \max_{P_X} \min_{W \in \overline{\cW_1 \cup \cW_2}} I(X;Y)$$
is a simple lower bound on $\Cdetcomm$. Recall that $\overline{\cW_1 \cup \cW_2}$ 
refers to the convex closure of the family of channels $\cW_1 \cup \cW_2$. 
Using the compound nature of the channel, this bound can be improved. 
In particular, we show the following.

\begin{thm} \label{thm:com}
(i) The random coding capacity for communication over \cavc\ is given by
\begin{equation} \label{eq:ca}
    \Crancomm = \max_{P_X} \min_{W \in \wu} I(X;Y).
\end{equation}
(ii) The deterministic capacity $\Cdetcomm > 0 $ if and only if the \cavc\ is not any-symmetrizable. If $\Cdetcomm > 0$, then 
\begin{equation} \label{eq:cdcom}
     \Cdetcomm = \Crancomm.
\end{equation} 
\end{thm}
\vspace{-10pt}

Refer to Section~\ref{sec:proofs} for proof sketches of Theorem~\ref{thm:com}.

\begin{figure}

 \center

  \includegraphics[width=0.4\textwidth]{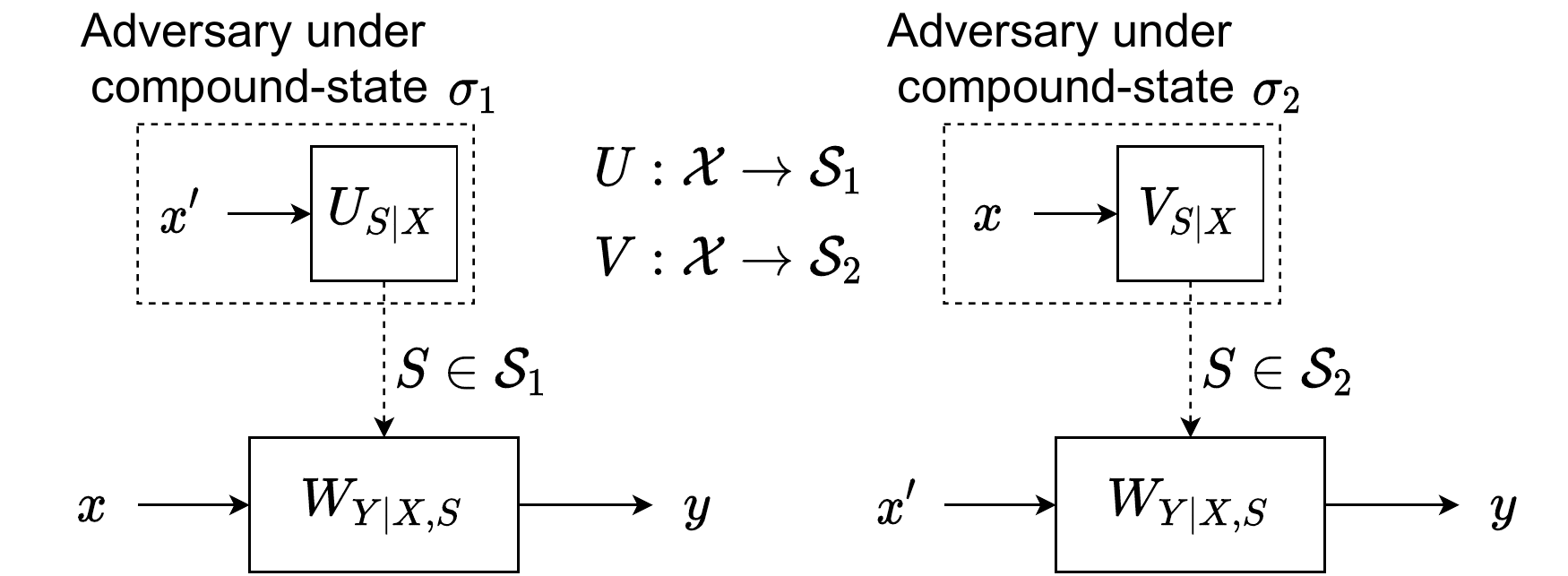}

  \caption{\textbf{Trans-symmetrizability:} 
  If there exists a pair of channels $U : \cX \to \cS_1$, $V : \cX \to \cS_2$ such that 
  the output distributions in the above two scenarios are the same for every 
  pair of symbols $(x,x')\in \cX^2$ then we call the channel trans-symmetrizable.}
  \label{fig:trans}

\end{figure}

%% file: model_both.tex
%\subsection{Joint Communication and Adversary Identification over \cavc\} \label{subs:both}

Let the deterministic capacity of the \cavc\ for the joint communication
and \cs\ identification be denoted by $\Cdetboth$ and let the random code
capacity be denoted by $\Cranboth$. Note that $\Cdetboth \leq \Cdetcomm$ as an
additional
constraint has been imposed in this problem. From Theorem~\ref{thm:com}, it is
clear that non-any-symmetrizability is required for joint communication and
\cs\ identification. Further, if $\ww \neq \emptyset$, then it is
possible for the adversary to emulate the channels in $\ww$ for either \cs. So it is not
possible to identify the \cs\ in such situations - this is true even under
random coding. Thus, $\ww = \emptyset$ is a necessary condition for joint
communication and \cs\ identification.

Any-symmetrizability and non-emptiness of $\ww$ are not implied by each other. 
This can be seen by the example satisfying $\cY_1 \cap \cY_2 = \emptyset$ in 
Section~\ref{subs:comm} and  the following example. 
Consider any non-symmetrizable AVC with state symbols in the set $\cS$. 
The \cavc\ with $\cS_1 = \cS_2 = \cS$ is not any-symmetrizable, but 
has $\ww \neq \emptyset$.

\begin{thm} \label{thm:both}
(i) The random coding capacity for joint communication and \cs\ identification over 
\cavc\  $\Cranboth=0$ if $\ww \neq \emptyset$. If $\ww = \emptyset$, then
\begin{equation} \label{eq:crboth}
    \Cranboth = \Crancomm.
\end{equation}
(ii) The deterministic capacity for joint communication and \cs\ identification 
$\Cdetboth > 0 $ if and only if the \cavc\ is not any-symmetrizable and $\ww = \emptyset$. 
If $\Cdetboth > 0$, then 
\begin{equation} \label{eq:cdboth}
     \Cdetboth = \Crancomm.
\end{equation}
\end{thm}
\vspace{-10pt}

%% file: model_auth.tex
Let the deterministic code capacity for the \cavc\ for the `communication or
\cs\ identification' problem be denoted by $\Cdetauth$ and the random code capacity
as $\Cranauth$. Observe that $\Cdetboth \leq \Cdetcomm \leq \Cdetauth$.  Since
the decoder needs to either communicate or identify the \cs , this is not
possible if the \cavc\ is trans-symmetrizable as trans-symmetrizability hinders
both the tasks of \cs\ identification and communication.
In Theorem~\ref{thm:auth}, we claim that non-trans-symmetrizability is necessary and sufficient for positive capacity 
- a significantly more relaxed condition as compared to non-any-symmetrizability.

\begin{remark} \label{rem:KK}
    %When $\cS_1\cap \cS_2\neq \emptyset$, and either adversary attacks
    %with a state from the intersection, it is not possible to detect the
    %adversary. The decoder thus needs to decode the message reliably.
    %In particular, 
    If $\wb \subseteq \wa$, then the decoder cannot identify the \cs\ $\sigma_2$
    reliably, and therefore, the decoder must recover the message in this case.
    %Under \cs\ $\sigma_1$, it may either decode the message or detect the adversary.
    The model in \cite{kosut-kliewer}
    considers an AVC (with state alphabet $\cS$) with a special
    no-adversary state $s_0 \in \cS$.  The decoder must decode the message
    correctly w.h.p. when the \avcstate\ sequence is $s_0^n = (s_0,\ldots,s_0)$. 
    For any other \avcstate\ sequence $\vecs \neq
    s_0^n$, the decoder may declare adversarial interference.
    This is a special case of our model with $\cS_2=\{s_0\}\subseteq \cS_1$.
    %
    %    This model is a generalization of the model in \cite{kosut-kliewer} where
    %the authors considers an AVC (with state alphabet $\cS$) with a special
    %no-adversary state $s_0 \in \cS$.  The decoder must decode the message
    %correctly w.h.p. under $s_0^n$. For any other state sequence $\vecs \neq
    %s_0^n$, the decoder may declare adversarial interference by outputting $\bot$.
    %\red{This is practically equivalent to a \cavc\ with $\cS_1 = \cS$ and $\cS_2 =
    %\{s_0\}$, $s_0 \in \cS$. Note that although the problem setting is not exactly
    %equivalent as in the present model, when \cs\ is $a_1$, the decoder still has
    %the choice of outputting the compound state identity but this is an impossible
    %task as explained in the paragraph above Theorem~\ref{thm:auth}; therefore, the
    %decoder has to communicate in this case.} 
\end{remark}    
For either \cs, consider the case when the adversary samples the \avcstate\ symbols independently and identically distributed (i.i.d.) according to $P_S$ such 
that $\sum_{s}P_S(s)W_{Y|X,S=s} \in \ww$. Here, the decoder cannot identify the \cs\ reliably, therefore the decoder
must recover the message.
% reliably to ensure a vanishing probability of error.
Thus, for any channel $W \in \ww$, the capacity of $W$ is an upper bound on 
$\Cdetauth$, i.e., $\Cdetauth \leq \max_{P_X} \min_{W \in \ww} I(X;Y)$. 
It is also possible to show that this upper bound is achievable when the \cavc\ is not trans-symmetrizable as described in Section~\ref{sec:proofs}.

\begin{thm} \label{thm:auth}
(i) The random coding capacity for `communication or \cs\ identification' 
over \cavc\ is given by
\begin{equation} \label{eq:cb}
    \Cranauth = \max_{P_X} \min_{W \in \ww} I(X;Y).
\end{equation}
In particular, if $\ww = \emptyset$, then $ \Cranauth = \infty$.\\
(ii) The deterministic capacity $\Cdetauth > 0 $ if and only if the \cavc\ 
is not trans-symmetrizable. If $\Cdetauth > 0$, then 
\begin{equation} \label{eq:crauth}
        \Cdetauth = \Cranauth.
\end{equation} 
\end{thm}
\vspace{-10pt}

If the \cs\ can be identified, then the message need not be decoded. So 
the capacity is infinite for such a \cavc. 
Thus, Theorem~\ref{thm:auth} implies that \cs\ can be identified  (i) under 
random coding if and only if $\ww = \emptyset$, and (ii) under deterministic 
coding if and only if the \cavc\ is not trans-symmetrizable and $\ww = \emptyset$.

\begin{corollary}
For a \cavc\ under deterministic coding, the \cs\ can be 
identified with arbitrarily small probability of error 
for sufficiently large block lengths if and only if the 
\cavc\ is not trans-symmetrizable and $\ww = \emptyset$.
\end{corollary}

% \textcolor{blue}{Even though no message is being decoded (only \cs\ is being identified), 
% the present achievability scheme of Theorem~\ref{thm:auth} requires a codebook of positive rate. 
% Sending a message uniformly from the message set adds stochasticity to the encoder.}

Note that for a non-trans-symmetrizable, but cis-symmetrizable \cavc\ with $\ww \neq \emptyset$, 
it is impossible to just communicate and it is impossible to identify the \cs\ separately; 
cis-symmetrizability hinders communication while $\ww \neq \emptyset$ hinders \cs\ identification. 
However, such channels would have a positive capacity according to Theorem~\ref{thm:auth} for the problem of 
`communication or \cs\ identification'.

%% file: section3.tex
\section{Proof Sketches} \label{sec:proofs}
We give a brief proof outline for the theorems. 
The full proofs can be found in Section~\ref{sec:4}. Let $\cP_k$ denote 
the set of all distributions over $\cS_k$, $k=1,2$.

%% file: proof_rand.tex
\subsection{Proof Sketch for Theorem~\ref{thm:com} (i)} %\label{subs:th2}

Both the achievability and converse parts of the proof follow along
similar lines as that for standard AVCs. The achievability argument uses
a randomly generated (and shared with the decoder) codebook where all code
 symbols are generated i.i.d. $\sim P_X$, a maximizing distribution of \eqref{eq:ca}.
%Please refer to the long version~\cite{long-version} for details.

%For any $R < \Crancomm$, choose $\delta>0$ be such that $R + \delta < \Crancomm$. 
%We only describe the encoder-decoder pair $(F_R,\Phi_R)$ (parameterized by the rate $R$) 
%used to achieve the capacity. The codebook for a $(2^{nR},n)$ code is be obtained by uniformly
%and independently sampling $M$ vectors ($\vecX_1,\ldots,\vecX_M$)  in $\tau_{X}, X\sim P_X$, $F_{R}(i) = \vecX_i$. 
%The decoder $\Phi_{R}(\vecy)=i \in \cM$ if there is a unique $i$ for which $I(X;Y) \geq R + \delta$ where $P_{XY} = P_{\vecX_i,\vecy}$, and $\Phi_{R}(\vecy)=1$ if no such $i$ exists. 

\subsection{Proof Sketch for Theorem~\ref{thm:both} (i)}

If $\ww \neq \emptyset$, then the adversary under either \cs\  can induce any
effective channel in $\ww$ using a suitable state distribution.
Thus the \cs\  cannot be identified reliably in this case.
The converse for the case $\ww = \emptyset$ follows from the converse of Theorem~\ref{thm:com}~(i). We now
outline the achievability argument under $\ww \neq \emptyset$.

For achievability, the encoder constructs
a vector with two parts $\vecx = (\hat \vecx, \tilde \vecx)$.
The first part is used for communication and the second part is used
for \cs\ identification. The vector $\vecx$ is randomly permuted before
transmission so that the adversary cannot apply different types of attack
on the two parts. The permutation is shared with the decoder, so that it
can recover $\vecx$. The encoding of the message in $\hat\vecx$ and its
decoding is similar to that in the proof of Theorem~\ref{thm:com}~(i). The
second part $\tilde \vecx$ is a fixed $|\cX|\log(n)$ length sequence
consisting of $\log(n)$ repetitions of each symbol in $\cX$.
The decoder estimates the effective channel law from this part and
identifies the \cs\ based on whether it is in \wa or in \wb.
The condition $\ww = \varnothing$ ensures that it is not in both \wa
and \wb.

\subsection{Proof Sketches for Theorem~\ref{thm:auth} (i)}

%As discussed after the definitions of cis-symmetrizability and
%trans-symmetrizability, communication at a non-zero rate is not possible if the
%channel is symmetrizable in either of these senses.
%As discussed in the proof 

For the converse proof, we first note that  since the adversary under either
\cs\ can induce a channel from $\ww$, the \cs\ cannot be
identified if the induced channel is in $\ww$. So the decoder must
decode the message reliably in such situation. However, by standard
arguments, the decoder cannot decode reliably if the rate is more than
$\Cranauth$.

We now discuss the achievability argument. The same coding scheme is used
as in Theorem~\ref{thm:both} (i) using a distribution $P_X$ that
maximizes \eqref{eq:cb}. If the effective channel induced (in
both $\tilde \vecx$ and $\hat \vecx $) %$\hat{\bx}$ and $\tilde{\bx}$) 
by the adversary is in $\ww$, then
the reliability in decoding follows using standard arguments since
the rate is less than $\min_{W \in \ww} I(X;Y)$. On the other hand,
if the effective channel is outside $\ww$, then the \cs\ can be identified,
as discussed in the proof of Theorem~\ref{thm:both} (i).

%% file: proof_det.tex
\subsection{Proof Sketches for Theorem~\ref{thm:com}~(ii), Theorem~\ref{thm:both}~(ii), 
Theorem~\ref{thm:auth}~(ii)}

It can be shown that $\Crancomm>0$ (resp. $\Cdetauth > 0$) 
when the channel is not any-symmetrizable (resp. trans-symmetrizable). 
%\begin{lemma}
%If the channel is not any-symmetrizable, then $\min_{W \in \wu} I(X;Y) > 0$ for all $P_X$ such that $P_X(x) > 0 \ \forall x \in \cX$. If the channel is
%not trans-symmetrizable, then $\min_{W \in \ww} I(X;Y) > 0$ for all $P_X$ such that $P_X(x) > 0 \ \forall x \in \cX$.
%\end{lemma} 
%\begin{proof}
%Suppose the statement is false, then there exists $P_X$ and $P_S \in \cP_1 \cup \cP_2$ for which $I(X;Y)=0$.
%Hence, there exists a distribution $P_{XSY} \in \cC_0$ such that $X$ and $Y$ are independent, i.e., $P_{Y|X}(y|x) = \sum_{s} W(y|x,s)P_S(s) = P_Y(y) \ \forall x,y$.
%The C-AVC is cis-symmetrizable in a trivial manner using $U(.|x) = P_S(.)$, a contradiction.
%\end{proof}
 
The achievability proof for deterministic coding follows  along similar lines of 
argument as in \cite{avc-csiszar-narayan}.
A suitable codebook with codewords $\vecx_1,\ldots,\vecx_M$ of type $P_X$ 
can be obtained using an extension of \cite[Lemma 3]{avc-csiszar-narayan}
for all the three theorems with appropriate $P_X$. We only describe the
decoders below, and refer the reader to Section~\ref{sec:4} for the detailed analysis. 
The decoder for the task of joint `communication and \cs\ identification' 
(Theorem~\ref{thm:both}~(ii)) is as described below.
Let 
\begin{equation}
    \mathcal{C}_{\eta} = \{  P_{XSY}: D(P_{XSY} || P_X \times P_S \times W) \leq \eta , \ P_S \in \cP_1 \cup \cP_2 \}.
\end{equation}
\begin{dec*} \label{dec:both-d}
Given codewords $\vecx_j$, $j=1,\ldots,M$, set $\pand(\vecy) = (i,\sigma_k)$, $i \in \cM, k\in\{1,2\}$, iff an $\vecs \in \cS_k^n$ exists such that:
\begin{enumerate}
    \item the joint type $P_{\vecx_i,\vecs,\vecy} \in \mathcal{C}_{\eta}$
    \item for each  $\vecx_j, j \neq i$ such that there exists $\vecs' \in \cS_{1}^n \cup \cS_2^n$, 
    $P_{\vecx_j,\vecs',\vecy} \in \mathcal{C}_{\eta}$, we have $I(XY;X'|S) \leq \eta$ where $P_{XX'SY} = P_{\vecx_i,\vecx_j,\vecs,\vecy}$.
\end{enumerate}
Set $\pand(\vecy)=(1,a_1)$ if no such $(i,a_k)$ exists.
\end{dec*}
The condition $\ww = \emptyset$ ensures that if there exists $\vecs \in \cS_1^n,\ P_{\vecx_m,\vecs,\vecy} \in \cC_{\eta}$ then 
$\forall \vecs' \in \cS_2^n,\ P_{\vecx_m,\vecs',\vecy} \notin \cC_{\eta}$.
For two distinct codewords $\vecx_i,\vecx_j$, and their corresponding $\vecs_i,\vecs_j$ respectively, 
(i) non-cis-symmetrizability ensures that they do not simultaneously satisfy both the 
decoder conditions when both $\vecs_i,\vecs_j \in \cS_k^n$ for some $k\in\{1,2\}$, 
(ii) non-trans-symmetrizability ensures they do not simultaneously satisfy both the decoder 
conditions when $\vecs_i \in \cS_k^n,\ \vecs_j \in \cS_{3-k}^n$ for some $k\in\{1,2\}$ (see Section~\ref{sec:4}). 

For Theorem~\ref{thm:com}, we can use a decoder similar to the above and disregard the decoder output corresponding to the \cs\ identity. 
For Theorem~\ref{thm:auth}~(ii), we show the achievability of  a non-zero rate, 
and then use the randomness reduction technique~\cite{csiszar-korner-book} 
to achieve the capacity. The following decoder is used to show positive capacity.
\begin{dec*} \label{dec:auth-d}
Given codewords $\vecx_j$, $j=1,\ldots,M$, let $B_k$ ($k=1,2$) be the set of messages $m \in \cM$ such that 
\begin{enumerate}
    \item the joint type $P_{\vecx_m,\vecs,\vecy} \in \mathcal{C}_{\eta}$
    \item for every $m'\neq m$ such that there exists  
    $\vecs' \in \cS_{3-k}^n$, $P_{\vecx_{m'},\vecs',\vecy} \in \mathcal{C}_{\eta}$, 
    we have $I(XY;X'|S) \leq \eta$ where $P_{XX'SY} = P_{\vecx_m,\vecx_{m'},\vecs,\vecy}$.
\end{enumerate}
If $B_1 = B_2 = \{m\}$, then $\por(\vecy)=m$. If for some $k\in\{1,2\}$,
$B_k=\emptyset \neq B_{3-k}$, then 
the decoder outputs the \cs\ $\por(\vecy)=\sigma_{3-k}$.
\end{dec*}
Non-trans-symmetrizability ensures that the two cases for $B_k$ 
described in the decoder are the only cases which can occur (see Section~\ref{sec:4}).

The rate-converses follow from the converse for the randomized coding capacity.
The zero-rate converse ideas have been discussed in Section~\ref{sec:results} and
are elaborated in Section~\ref{sec:4}.

%% file: section4.tex
\section{Complete Proofs} \label{sec:4}

We use the notation $W_P$ to refer to the channel $W_P:\cX \to \cY$
given by $\sum_s P(s)W_{Y|X,S=s}$.
The $\epsilon$-typical set of a random variable
be denoted by $\tau_X^{\epsilon} =
\{\vecx:|P_{\vecx}(x)-P_X(x)|\leq \epsilon\  \forall x \in \cX\}$. In
particular, $\tau_X$ denotes the typical set when $\epsilon=0$.
Let $\pn{L}{n}$ denote the set all emirical distributions of length $n$ 
over the set $L$.

%% file: ran_conv.tex
\subsection{Converse Proofs Under Random Coding} \label{subs:ran-conv}
\renewcommand\qedsymbol{$\blacksquare$}
\begin{lemma} \label{lem:rcom-conv}
$$ \Crancomm \leq \max_{P_X} \min_{W \in \wu} I(X;Y) $$
\end{lemma} 
\begin{proof}
Consider the adversarial strategy for \cs\ $\sigma_k$ where 
the adversary chooses a distribution $P(\vecs)$ with support over $\cS_k^n$ and randomly samples a 
vector $\vecs$ distributed according to $P$. 
Note that the \cavc\ average error probability under the worst-case $P$ is same as that under worst-case $\vecs$
(\textit{c.f.} \cite[Lemma 12.3, Page 210]{csiszar-korner-book}). In other words, if $\cP_i^{(n)}$ represents all 
distributions over $\cS_k^n$, then
$$ P_e^r(Q,k) =  P_e^p(Q,\cP_k^{(n)}), $$
where 
$$ P^p_e(Q,\cP_k^{(n)}) = \max_{P \in \cP_k^{(n)}}\sum_{\vecs}P(\vecs) \sum_{(f,\phi)}
    Q(f,\phi) \frac{1}{M} \sum_{m=1}^M W^n(\phi^{-1}(\cE_{m,k}) |f(m),\vecs).  $$

Here, $\cE_{m,k}$ is the error event corresponding to communication error $\{m' \in \cM:m' \neq m \}$.

Consider a particular class of adversarial strategies for \cs\ $\sigma_k$ where the adversary chooses the state 
sequence $\vecs$ with each bit independently from the distribution $ P_k \in \cP_k$, i.e., $P(\vecs) = P_k^n(\vecs) = \prod_{j=1}^n P_k(s_j)$.
The probability of error under this adversarial stragey is given by 

$$ \sum_{(f,\phi)} Q(f,\phi) \frac{1}{M} \sum_{m=1}^M W^n_{P_k}(\phi^{-1}(\cE_{m,k}) |f(m)), $$
where $W_{P_k}(y|x) = \sum_{s \in \cS_k} P_k(s)W(y|x,s)$.  
Therefore, channel distribution is given by Discrete Memoryless Channel (DMC) $W_{P_k}$. \\

Under such i.i.d. adversarial strategy, consider a sequence of codes with rate $R'$
such that the error probability $P_e^{(n)}$ tends to $0$ for large block-length. 
Let $M$ be the message which is encoded into vector $X^n$ and transmitted,  and let 
$Y^n$ be the vector received by the decoder. 
Then, $(M,Y^{i-1}) \longleftrightarrow X_i \longleftrightarrow Y_i$
form a Markov Chain under this adversarial strategy (as $W_{P_k}$ is a DMC). 
Let $\hat M$ be the decoded message.
By Data-Processing and Fano's inequalities,
 $$ H(M|Y^n) = H(M|\hat M) \leq 1+P_e^{(n)}nR'=n\epsilon_n, $$
where $\epsilon_n$ is defined as $\frac{1}{n}+P_e^{(n)}R'$. 
Next, we note that
\begin{align}
    nR' &= H(M) \\
        &= H(M|Y^n) + I(M;Y^n) \\
        &\leq n\epsilon_n + I(M;Y^n).   \label{eq:fano}
\end{align} 

Consider the term $I(M;Y^n)$ - 
\begin{align}
    I(M;Y^n) &= \sum_{j=1}^n I(M;Y_j|Y^{j-1}) \\
            &\leq \sum_{j=1}^n I(M,X_j,Y^{j-1};Y_j) \\
            &= \sum_{j=1}^n I(X_j;Y_j), \\
\end{align}
where the last equality follows from the property of Markov Chains 
($(M,Y^{i-1}) \longleftrightarrow X_i \longleftrightarrow Y_i$). \\

Let $L\sim Uniform[1,n]$ be independent of other random variables. 
Note that $L \longleftrightarrow X_L \longleftrightarrow Y_L$ forms a Markov Chain.
Thus, we have,
\begin{align}
   \frac{1}{n}\sum_{j=1}^n I(X_j;Y_j) &= I(X_L;Y_L|L) \\
   &\leq I(X_L,L;Y_L) \\
   &= I(X_L;Y_L).
\end{align}
Since \eqref{eq:fano} has to hold for all such i.i.d. adversarial strategies,
$$ R' \leq \epsilon_n +  \min_{P \in \cP_1 \cup \cP_2} I(X;Y), $$
where  $Y$ is related to $X$ via the DMC $W_{P}$. 
Further, $\epsilon_n$ can be made arbitrarily small by choosing $n$ large enough 
since $P_e^{(n)}$ vanishes for large $n$.
Therefore, for every achievable rate $R'<\Crancomm$, we have,
$$ \Crancomm  \leq \max_{P_X} \min_{W \in \wu} I(X;Y). $$
\end{proof}

\vspace{20pt}

Formally, we define the task of only \cs\ identification (without requiring reliable message decoding). 
Let $\hat \cM = \{\sigma_1,\sigma_2\}$ and
define $ \hat \cE_{k} := \left\{\sigma_{3-k} \right\} $ (similar to $\cE_{m,k}$ defined in 
Section~\ref{sec:2}).
Denote the probability of error in \cs\ identification as $P_{id}^r(Q)$ which is described in terms
of $P_{id}^r(Q,k)$ as  
\begin{align}
    P_{id}^r(Q,k) &\triangleq \max_{\vecs \in \cS_k^n} \sum_{(f,\phi)} Q(f,\phi) \frac{1}{M} \sum_{m=1}^M W^n(\phi^{-1}(\hat \cE_{k}) |f(m),\vecs), \text{ and}\\  
    P_{id}^r(Q) &= \max\{ P_{id}^r(Q,1), P_{id}^r(Q,2) \label{eq:def} \}. 
\end{align}

We first show that $\ww = \emptyset$ is necessary for \cs\ identification, which also implies that it is 
necessary for simultaneous \cs\ identification and communication. 

\begin{lemma} \label{lem:ad-id}
$\ww = \emptyset$ is necessary for \cs\ identification under random coding.
\end{lemma}
\begin{proof}
Let $\ww  \neq \emptyset $, then $\exists$ channel $Z:\cX \to \cY$ , $Z_{Y|X} \in \ww$. 
Therefore, we can choose distribution $P_k$ over $\cS_k$ such that 
$Z_{Y|X} = \sum_{s}  P_k(s)W_{Y|X,S=s}$ for $k=1,2$. \\
Let $T_k(\vecs) := \prod_{j=1}^n P_k(s_j)$. Consider an adversarial
stragey where the adversary chooses the state i.i.d.  from distribution $P_k$
when the compound state is $\sigma_k$. Under this attack and compound state $\sigma_k$, 
we have,  
%%%
\begin{align}
  P_{id}^r(Q,k) &\geq  \sum_{\vecs} T_k(\vecs) \sum_{(f,\phi)} Q(f,\phi) \frac{1}{M} \sum_{m=1}^M W^n(\phi^{-1}(\hat\cE_{k}) |f(m),\vecs) \\
    &=   \frac{1}{M} \sum_{m=1}^M \sum_{\vecs} \sum_{(f,\phi)} \sum_{\vecy \in \phi^{-1}(\hat\cE_k)} Q(f,\phi) T_k(\vecs)  W^n(\vecy |f(m),\vecs) \\
    &=  \frac{1}{M} \sum_{m=1}^M \sum_{\vecs} \sum_{(f,\phi)} \sum_{\vecy \in \phi^{-1}(\hat\cE_k)} Q(f,\phi) \prod_{j=1}^n T_k(s_j)  W^n(y_j |f(m)_j,\vecs_j) \\
    &=  \frac{1}{M} \sum_{m=1}^M  \sum_{(f,\phi)}  Q(f,\phi) Z^n(\phi^{-1}(\hat\cE_k) |f(m)).
\end{align}
Hence,
\begin{align}
    P_{id}^r(Q,1)+P_{id}^r(Q,2) &\geq  \frac{1}{M} \sum_{m=1}^M  \sum_{(f,\phi)}  Q(f,\phi) Z^n(\phi^{-1}(\hat\cE_1) \cup \phi^{-1}(\hat\cE_2)|f(m)) \\
     &\geq  1 \hspace{1em} \forall Q, \label{eq:conv}
\end{align}
where \eqref{eq:conv} follows as $\hat\cE_1 \cup \hat\cE_2 = \cY^n$. 
Therefore, \cs\ identification is not possible if $\ww \neq \emptyset$.

\end{proof}
\vspace{20pt}

Note that the probability of error in only \cs\ identification is strictly less than or equal to the probability of 
error in joint \cs\ identification and communication. Thus,
if the error probability in \cs\ identification is not vanishing for a \cavc , then the error probability 
in joint communication and \cs\ identification cannot vanish. 
Lemma~\ref{lem:ad-id} establishes that $\Cranboth = 0$ if $\ww \neq \emptyset$. 
If $\ww = \emptyset$, the fact $\Cranboth \leq \Crancomm$ and Lemma~\ref{lem:rcom-conv} establish that 
$$\Cranboth \leq \max_{P_X} \min_{W \in \wu} I(X;Y)$$.

\begin{lemma} \label{lem:crauth-conv}
\begin{equation}
    \Cranauth \leq \max_{P_X} \min_{W \in \ww} I(X;Y) \label{eq:crauth-conv}
\end{equation}
\end{lemma}
\begin{proof}
When $\ww = \emptyset$, RHS of \eqref{eq:crauth-conv} is infinity and the relation holds trivially.  

If $\ww \neq \emptyset$, then let $Z_1,\ldots,Z_n$ be any $n$ channels $\in \ww$.
We represent the $n$-length channel as $Z^{(n)}(\vecy|\vecx) = \prod_{i=1}^n Z_i(y_i|x_i)$. 
Define $P_{i,k}(s) \in \cP_k$ for $k=1,2$ such that $\sum_s P_{i,k}(s)W_{Y|X,S=s} = Z_i$. 
We have,
%Let $R$ be an achievable rate such that $P_e^r(Q) \to 0$ as $n \to  \infty$.
\begin{align}
        P_e^r(Q,k)  &\geq \sum_{\vecs \in \cS_k^n} \prod_{i=1}^n P_{i,k}^n(s_i) \sum_{(f,\phi)}Q(f,\phi) \frac{1}{M} 
                            \sum_{m=1}^M W^n(\phi^{-1}(\cE_{m,k})|f(m),\vecs) \\
                    &= \frac{1}{M} \sum_{(f,\phi)} \sum_{m=1}^M Q(f,\phi) Z^{(n)}(\phi^{-1}(\cE_{m,k})|f(m)); \\
\implies 2P_e^r(Q)  &\geq  \frac{1}{M} \sum_{(f,\phi)} \sum_{m=1}^M Q(f,\phi) Z^{(n)}(\phi^{-1}(\cE_{m,1} \cup \cE_{m,2})|f(m)) \\
                    &=  \frac{1}{M} \sum_{(f,\phi)} \sum_{m=1}^M Q(f,\phi) Z^{(n)}(\phi^{-1}(\{\sigma_1,\sigma_2\}\cup \cM \setminus m)|f(m)) \\
                    &=  \frac{1}{M} \sum_{(f,\phi)} \sum_{m=1}^M Q(f,\phi) Z^{(n)}(\phi^{-1}(m)^C|f(m)). 
\end{align}
In order to get $P_e^r(Q) \to 0$, we must ensure the RHS vanishes as $n$ increases for all $Z^{(n)}$ with $Z_i \in \ww$. 
The RHS is exactly the probability of error for communication over an AVC with the family of channels $\ww$. Thus, we have,
$$ \Cranauth \leq \max_{P_X} \min_{Z \in \ww} I(X;Y)$$
\end{proof}
\vspace{20pt}

%% file: ran_ach.tex
\subsection{Achievability Proof of Theorem~\ref{thm:com} (i)} \label{subs:ran-ach}

\begin{lemma} \label{lem:crcom-ach}
$$\Crancomm \geq \max_{P_X} \min_{W \in \wu} I(X;Y)$$
\end{lemma}

The proof is along the lines of \cite[Lemma 5]{kosut-kliewer}.
For any $R < \Crancomm$, choose $\delta>0$ such that $R + \delta < \Crancomm$. 
We describe the encoder-decoder pair $(F_R,\Phi_R)$ (parameterized by the rate $R$) 
used to achieve the capacity. The codebook for a $(2^{nR},n)$ code is obtained by uniformly
and independently sampling $M$ vectors ($\vecX_1,\ldots,\vecX_M$)  $\in \tau_{X}$ where 
$\tau_{X}$ is the typical set corresponding to some $P_X \in \pn{\cX}{n}$, and $F_{R}(i) = \vecX_i$. 
The decoder outputs $\Phi_{R}(\vecy)=i \in \cM$ if there is a unique $i$ for which 
$I(X;Y) \geq R + \delta$ where $P_{XY} = P_{\vecX_i,\vecy}$, and $\Phi_{R}(\vecy)=1$ if no such $i$ exists. 

If message $i$ is sent and the \avcstate\ sequence is $\vecs$ during transmission, 
we need to prove the following two results to show that rate $R$ is achievable :
\begin{align}
    \bbP\{ (\vecX_i,\vecy) \in \tau_{XY}, I(X;Y) < R + \delta  \} & \xrightarrow{\: n \to \infty \: } 0 \ \ \forall \vecs \in \cS_1^n \cup \cS_2^n, \text{ and} \label{eq:part-1} \\
    \bbP\{ (\vecX_j,\vecy) \in \tau_{XY}, I(X;Y) \geq R + \delta, \text{for some } j \neq i \} &\xrightarrow{\: n \to \infty \: } 0 \ \ \forall \vecs \in \cS_1^n \cup \cS_2^n. \label{eq:part-2} 
\end{align}
The probability expression in the LHS of \eqref{eq:part-1} is equal to
\begin{align}
    & \sum_{\substack{P_{XSY}:I(X;Y)<R + \delta, \\ \vecs \in \tau_S}} \  \sum_{\vecx \in \tau_{X|S}(\vecs)} |\tau_X|^{-1} \sum_{\vecy \in \tau_{Y|XS}(\vecx,\vecs)} W^n(\vecy|\vecx,\vecs)\\
    \leq& \sum_{\substack{P_{XSY}:I(X;Y)<R + \delta, \\ \vecs \in \tau_S}} \  \sum_{\vecx \in \tau_{X|S}(\vecs)} |\tau_X|^{-1} \exp\{-nD(P_{XSY}||P_{XS} \times W) \} \\ 
    =& \sum_{\substack{P_{XSY}:I(X;Y)<R + \delta, \\ \vecs \in \tau_S}} \  \frac{|\tau_{X|S}(\vecs)|}{|\tau_X|} \exp\{-nD(P_{XSY}||P_{XS} \times W) \} \\
    \leq& \sum_{\substack{P_{XSY}:I(X;Y)<R + \delta, \\ \vecs \in \tau_S}} \   \exp\{-n(D(P_{XSY}||P_{XS} \times W) + I(X;S) - \epsilon )\}. \label{eq:poly}
\end{align}
Using the fact that  $D(P_{XSY}||P_{XS} \times W) + I(X;S) = D(P_{XSY}||P_{X} \times P_S \times W)$, and taking the marginals  along $\cX \times \cY$, 
while noting that divergence does not increase with marginalization, we have
$$ \bbP\{ (\vecX_i,\vecy) \in \tau_{XY}, I(X;Y) < R + \delta  \} \leq \sum_{\substack{P_{XSY}:I(X;Y)<R + \delta, \\ \vecs \in \tau_S}} \   \exp\{-n(D(P_{XY}||P_X \times W_{P_S}) - \epsilon )\}, $$
where $W_{P_S} = \sum_s P_S(s)W_{Y|X,S=s}$. 
In \eqref{eq:poly}, we can set $\epsilon$ arbitrarily small as $\epsilon$ is present to account for the $(n+1)^{|\cX|}$ term which 
grows polynomially. In particular, set $\epsilon < \epsilon'$, where $\epsilon'$ is described next. \\
Note that if $P_{XY} = P_X \times W_{P_S}$, then $R+\delta < I(X;Y)$ (as $W_{P_S} \in \wu$) by choice of $R$ 
and $\delta$ as described. Since mutual information and relative entropy are continuous functions of $P_{XY}$, 
there exists $\epsilon'>0$ such that if $I(X;Y)<R+\delta$, then
$$ D(P_{XY}||P_X \times W_{P_S}) \geq \epsilon' \ \forall P_S, \text{ or equivalently, }\forall \vecs. $$
Since there are only polynomially many types, for sufficiently large $n$, 
\eqref{eq:part-1} is less than $\exp\{-n(\epsilon'-\epsilon)/2\} \to 0$ as $n \to \infty$.  
\\

Next, we analyze the probability in the LHS of \eqref{eq:part-2}. The probability, for any $\vecs$, can be written as
\begin{align}
      &=    \sum_{\substack{P_{XX'SY}:I(X';Y)\geq R+\delta \\ \vecs \in \tau_{S}, \ P_X = P_{X'}}} \ \sum_{\vecx_i \in \tau_{X|S}(\vecs)} |\tau_X|^{-1}
            \sum_{j=1,j\neq i}^M \ \sum_{\vecx_j \in \tau_{X'|XS}(\vecx,\vecs)} |\tau_X|^{-1} 
            \sum_{\vecy \in \tau_{Y|XX'S}(\vecx_i,\vecx_j,\vecs) } W^n(\vecy|\vecx_i,\vecs) \\
    &\leq \sum_{\substack{P_{XX'SY}:I(X';Y)\geq R+\delta \\ \vecs \in \tau_{S}, \ P_X = P_{X'}}} \exp(-n(I(X;S)-\epsilon))
          \exp(nR) \exp(-n(I(X';XS)-\epsilon)) \exp(-n(I(Y;X'|XS)-\epsilon)) \\
    &\leq \sum_{\substack{P_{XX'SY}:I(X';Y)\geq R+\delta \\ \vecs \in \tau_{S}, \ P_X = P_{X'}}} \exp\{ -n(I(X;S)+I(X';XSY)-R -3\epsilon)\} \\
    &\leq \sum_{\substack{P_{XX'SY}:I(X';Y)\geq R+\delta \\ \vecs \in \tau_{S}, \ P_X = P_{X'}}} \exp\{ -n(I(X';Y)-R-3\epsilon)\} \\
    &\leq \sum_{\substack{P_{XX'SY}:I(X';Y)\geq R+\delta \\ \vecs \in \tau_{S}, \ P_X = P_{X'}}} \exp\{- n(\delta-3\epsilon)\} \\
    &\leq \exp\{- n(\delta-3\epsilon-\epsilon').\} 
\end{align}
Note that $\epsilon$ and $\epsilon'$ can be set arbitrarily small as they are present to account for polynomially many terms. 
This proves the achievability of the capacity $\Crancomm$.

\subsection{Achievability Proof of Theorem~\ref{thm:both} (i) and Theorem~\ref{thm:auth} (i)} \label{subs:ran-ach-2}

We begin this sub-section by focusing on identifying the \cs\ under random coding as the method discussed would be directly used
for proving achievability for Theorem~\ref{thm:both} and Theorem~\ref{thm:auth}. We present the following 2 lemmas 
before describing \cs\ identification.

\begin{lemma} \label{lb-1}
In a \cavc , let the random vector $\vecX$, chosen uniformly from  the typical set $\tau_X$ corresponding to some distribution $P_X \in \pn{\cX}{n}$,
be the input and the \avcstate\ sequence be $\vecs \in \mathcal{S}_k^n$. Suppose $\vecY$ represents the output sequence. 
%then $(\vecX,\vecY) \sim   |\tau_{X}|^{-1} W^n(\vecy|\vecx,\vecs),\ \vecx \in \tau_X $. 
Then, for any $\epsilon>0$ and sufficiently large $n$, the joint type $(\vecX,\vecY) \in \tau_{XY}^{\epsilon}$ 
with high probability, where $\tau_{XY}^{\epsilon}$ is the typical set corresponding to the distribution 
$P_{XY} = P_X \times \widetilde Z_{Y|X}$,  for some $\widetilde Z_{Y|X} \in \wk$.
\end{lemma}

The proof for Lemma~\ref{lb-1} can be found in the Appendix.

\begin{lemma} \label{lb-2}
If $\ww = \emptyset$ then for any $Z:\cX \to \cY, Z_{Y|X} \in \wa$, 
any $V:\cX \to \cY,V_{Y|X} \in \wb$, and any distribution $P$ over $\cX$ 
such that $P(a) > 0,\ \forall a \in \mathcal{X}$, there exists some $\eta>0$  such that
$$\sup_{(a,b) \in \cX\times \cY } \{ |P(a)Z_{Y|X}(b|a) - P(a)V_{Y|X}(b|a)| \} > \eta.$$
\end{lemma} 

In fact, instead of just $\wa$ and $\wb$, Lemma~\ref{lb-2} holds for any 
two closed and disjoint sets of channels.

\begin{lemma} \label{lem:ad-id-ach}
    $\ww = \emptyset$ is sufficient for \cs\ identification under random coding.
\end{lemma}
\begin{proof}
    Refer to equation \eqref{eq:def} for definition of probability of error in the \cs\ identification task.
    In this setting, there is no particular need or meaning in sending any `message' since
    the decoder does not even try to decode the message. However,
    since there is a message term used in the error probability definition in \eqref{eq:def}, we still need
    to describe the encoder in terms of messages.
    For our achievability scheme, consider an encoder which randomly samples a vector from $\vecF \in \tau_X$ 
    (for some distribution $P_X \in \pn{\cX}{n}$)
    and for each message, it outputs the same vector $\vecF$, i.e., for any realisation of the encoder, 
    the output is same for all the messages (this form of degenrate encoder is sufficient for proving the lemma). 
    Since the decoders knows which encoder is used (shared randomness), 
    it knows the exact vector which is transmitted by the encoder. 
    Represent the encoder output as $F(i) = \vecF \in \tau_{X} \ \forall i \in \cM$. 
    
    \begin{dec*}
    $G(\vecy) = \sigma_k$ if $\exists$ $Z_{Y|X} \in \wk$ such that 
    $(\vecF,\vecy) \in \tau_{XY}^{\epsilon}$ for $P_{XY} = P_X \times Z_{Y|X} $ and there exists no such $Z_{Y|X} \in \overline{\cW}_{3-k}$. \\ 
    Else arbitrarily set $G(\vecy) = \sigma_1$.
    \end{dec*}
    We sepcify $\epsilon$ later in this proof.
    
    Probability of error in identification for the encoder-decoders described is given by
    \begin{align}
    P_{id}^r(Q,k) &= \max_{\vecs \in \cS_k^n} \sum_{\vecf} |\tau_{X}|^{-1} \frac{1}{M} \sum_{m=1}^M W^n(\Phi^{-1}(\hat \cE_{k}) |\vecf,\vecs) \\
     &= \max_{\vecs \in \cS_k^n} |\tau_{X}|^{-1} \sum_{\vecf}  W^n(\Phi^{-1}( \cE_{k}) |\vecf,\vecs).
    \end{align}
    
    The error event $\hat \cE_k$ can be due to 2 events - \\
    (A) When no such $Z_{Y|X} \in \wk$ such that $(\vecf,\vecy)$ is in the typical set. \\
    (B) When there is a $V_{Y|X} \in \overline{\cW}_{3-k}$ such that $(\vecf,\vecy)$ is in the typical set. \\ \\
    For each $\vecs$, we now analyze these 2 cases. \\
    \textbf{(A)}:\\
    By choosing $ Z_{Y|X}$ as defined in Lemma~\ref{lb-1}, for any $\epsilon$ and sufficiently large $n$, 
    the probability of this event can be made arbitrarily small. \\ \\
    \textbf{(A)$^C \cap$ (B)}:\\
    The event $(A)^C \cap (B)$ implies $\exists V_{Y|X} \in \overline{\cW}_{3-k}$ such that 
    $(\vecf,\vecy) \in \tau_{XY}^{\epsilon}$ for $P_{XY}=U_X\times V_{Y|X}$ and 
    $\exists Z_{Y|X} \in \wk$ such that $(\vecf,\vecy) \in \tau_{XY}^{\epsilon}$ for $P_{XY}= U_X\times  Z_{Y|X}$.
    Therefore,
    $$ |U(a) Z_{Y|X}(b|a) - U(a) V_{Y|X}(b|a) |<2\epsilon \  \forall (a,b) \in \mathcal{X}\times\mathcal{Y}.$$
    We can choose sufficiently small $\epsilon$ such that $\epsilon < \eta/2$ which would violate Lemma~\ref{lb-2}, 
    implying that this case occurs with arbitrarily low probability.  \\
    Hence, $P_{id}^r(Q,k)$ can be made arbitrarily small for large $n$. 
    Thus, we can identify the \cs\ under random coding as stated in the theorem when $\ww = \emptyset$. 
    
\end{proof}

For achievability of both Theorem~\ref{thm:both} (i) and Theorem~\ref{thm:auth} (ii), we use a similar encoding scheme.
Let $\tilde \vecx$ be an $|\cX|\log(n)$ length sequence consisting of $\log(n)$ repitions of each symbol in $\cX$. 
For Theorem~\ref{thm:both} (i), a $(2^{nR'},n')$, code $(F^{\mathsf{and}},\Pand)$ consists of a 
length-$n$ communication part and length $n'-n$ \cs\ identification part where $n$ is such that $n' = n + |\cX |\log(n)$. 
The communication part of a code is given in terms encoder of Lemma~\ref{lem:crcom-ach} $F_R$, $R = \frac{R'n'}{n}$, 
and the indetification part consists of the constant vector $\tilde \vecx$ as shown in Figure~\ref{fig:vec}.
Let $\Gamma$ be a random and uniformly choosen permutation of length $n'=n+|\cX|\log(n)$.
The encoder $F^{\mathsf{and}}(i) = \Gamma(F_{R}(i),\tilde \vecx),\ i \in \{1,\ldots,2^{nR}\}$. 
Note that the rate $R'=\frac{Rn}{n'}$ of the code is governed by $R$ for large block length.
For Theorem~\ref{thm:auth} (i), we use the same structure of the encoder but operate at a different rate $R'$. 
The encoder of a $(2^{nR'},n')$, code $(F^{\mathsf{or}},\Por)$ is given by 
$F^{\mathsf{or}}(i) = \Gamma(F_{R}(i),\tilde \vecx),\ i \in \{1,\ldots,2^{nR}\}$ ($R',R$ is different for 
$F^{\mathsf{and}}$ and $F^{\mathsf{or}}$).

\begin{figure}[ht]
    \centering
    \includegraphics[width=0.5\textwidth]{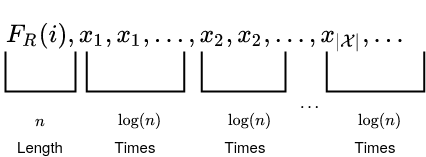}
    \caption{The vector $(F_R(i),\tilde \vecx)$}
    \label{fig:vec}
\end{figure}

Due to the shared randomness, the decoder knows the realisation of $F_R$ and $\Gamma$.
The decoder uses $\Gamma$ to get back the original ordering, i.e., to get $(\hat \vecy,\tilde \vecy) = \Gamma^{-1}(\vecy)$. 
Here, $\hat \vecy$ represents the vector corresponding to the first $n$ symbols and 
$\tilde \vecy$ represent the vector corresponding to the last $|\cX|\log(n)$ symbols of $\Gamma^{-1}(\vecy)$.  
If the \avcstate\ sequence during transmission is represented as $\vecs$, 
then let $\vecs_a = [\Gamma^{-1}(\vecs)]^{n}_1$ and $\vecs_b = [\Gamma^{-1}(\vecs)]_{n+1}^{n'}$ - this notation
is explained in the footnote\footnote{For a sequence $\vecy$, we use the notation $[\vecy]_a^b,\ (b>a)$ to
refer to the subsequence $(y_a,\ldots,y_b)$.}.

\begin{lemma} \label{lem:crboth-ach}
    When $\ww = \emptyset$,
    $$\Cranboth \geq \max_{P_X} \min_{W \in \wu} I(X;Y).$$
\end{lemma}
\begin{proof}
We use the encoding scheme described above and use $\Gamma^{-1}$ at the decoder $\Pand$, i.e., 
the decoder obtains $(\hat \vecy,\tilde \vecy) = \Gamma^{-1}(\vecy)$. By the method described in Lemma~\ref{lem:ad-id-ach}, 
one can identify the \cs\ as $G(\tilde \vecy)$ (with $\vecF$ in the lemma being the vector $\tilde \vecx$) correctly w.h.p. for large block length. 
Note that this encoding scheme of shuffling $\tilde \vecx$ is equivalent to sending a vector from the typical set 
of the uniform distrbituion over $\cX$ described in Lemma~\ref{lem:ad-id-ach}. \\

For any $R = \frac{n'R'}{n} < \max_{P_X} \min_{W \in \wu} I(X;Y)$, we use the same decoder $\Phi_R$ used in Lemma~\ref{lem:crcom-ach}
to decode the message. We obtain the message $\hat m = \Phi_R(\hat \vecy)$ correctly w.h.p. 
Thus, using the $(2^{n'R'},n')$ code, we can communicate at rate $R' = \frac{nR}{n'}$. For large block length, $R' \to R$.
\end{proof}

We now focus on proving achievability of Theorem~\ref{thm:auth}. We present two lemmas before going into 
the main proof. The following Lemma is a well-known result and can be found in \cite{skala2013hypergeometric}.
\begin{lemma} \label{lem:urn}
An urn contains $M$ white balls and $N-M$ black balls. If $n$ balls are drawn uniformly without replacement 
and $i$ represents the number of white balls drawn then, $\bbE[i] = n\frac{M}{N}$. 
Further, we can bound the deviations from the mean as shown,
\begin{align}
    \bbP[i \geq \bbE[i] + tn] &\leq e^{-2t^2n}, \\
    \bbP[i \leq \bbE[i] - tn] &\leq e^{-2t^2n}, \\
    \bbP[ |i - \bbE[i]| \geq tn] &\leq 2e^{-2t^2n}. 
\end{align}
\end{lemma}

Using Lemma~\ref{lem:urn}, we obtain the following.
\begin{lemma} \label{lem:type}
    Let random variable $S$ be distributed as $P_{\vecs}$. Then
    \begin{align}
        \bbP( \vecs_a \notin \tau_S^{\eta}) &\leq 2\max\{|\cS_1|,|\cS_2|\}n^{-2\eta^2|\cX|} \ ,\\
        \bbP( \vecs_b \notin \tau_S^{\eta}) &\leq 2\max\{|\cS_1|,|\cS_2|\}e^{-2\eta^2n}.
    \end{align}
\end{lemma}
\begin{proof}
    Since $\Gamma$ shuffles randomly and uniformly, this follows directly from the definition of typicality 
    and Lemma~\ref{lem:urn}. The $\max\{.\}$ operator is present to ensure that the inequality is valid when $\vecs$ belongs 
    to either of the two \cs .
\end{proof}
\vspace{20pt}
    
Lemma~\ref{lem:type} shows that the \avcstate\ sequence vector corresponding to the identification part 
and the communication part have roughly the same type as the entire vector $\vecs$. 

\begin{lemma}
\begin{equation}
    \Cranauth \geq \max_{P_X} \min_{W \in \ww} I(X;Y) \label{eq:crauth-ach}
\end{equation}
\end{lemma}
\begin{proof} 
We use the encoding scheme described after Lemma~\ref{lem:crcom-ach}. 
We specify the rate $R$ of communication corresponding to the communication part later. 
Let the encoder-decoder pair for the $(2^{n'R'},n')$, $R' = nR/n'$ code be $(F^{\mathsf{or}},\Por)$. 
Note that if $\ww = \emptyset$ then we can use the adversary identification scheme as described in 
Lemma~\ref{lem:ad-id-ach} to achieve infinite capacity using $(\tilde \vecx, \tilde \vecy)$. 
If $\ww \neq \emptyset$, then we first use a 
communication decoder $\tilde \Phi_R: \cY^n \to \cM \cup \bot$ described below to decode the message.

The codebook for a $(2^{n'R'},n')$ code is obtained by uniformly
and independently sampling $M=2^{n'R'}$ vectors ($\vecX_1,\ldots,\vecX_M$)  in $\tau_{X}$ with some $P_X \in \pn{\cX}{n}$ 
and $F^{\mathsf{or}}(i) = (\vecX_i,\tilde \vecx)$. 
The decoder outputs $\tilde \Phi_{R}(\vecy)=i \in \cM$ if there is a unique $i$ for which $I(X;Y) \geq R + \delta$ where $P_{XY} = P_{\vecX_i,\vecy}$, 
and $\tilde \Phi_{R}(\vecy)= \bot$ if no such $i$ exists.

We show that the communication decoder correctly decodes the message w.h.p. (with high probability) for a certain class of adversarial attacks. 
For other attacks, we show that the decoder may output the correct message or output $\bot$ but it would not decode 
to a wrong message w.h.p. On receiving an error ($\bot$), a second decoder - \cs\ decoder - would be used to identify the \cs . 

Suppose the \cs\ is $\sigma_k$ and the adversary operates with \avcstate\ sequence $\vecs \in \cS_k^n$. 
Let dummy random variable $S\sim P_{\vecs}$. Let $|| P_X ||$ denote the max norm of a distribution - $\max_x P_X(x)$. \\
We use $\hat \vecy$ (defined in the text following Lemma~\ref{lem:ad-id-ach}) and $\tilde \Phi_R$ for decoding the message. \\
Define the set $\cP_0 = \{P \in \cP_1 \cup \cP_2: W_{P} \in \ww \}$. 
Let $\cP^+_{\epsilon} = \{P \in \cP_1 \cup \cP_2: \exists P' \in \cP_0, \ ||P-P'|| \leq \epsilon \}$ 
and let $\cW^-_{\epsilon} = \{W_P: P \in \cP^+_{\epsilon} \}$. 
Also, define  $\cW^+_{\epsilon} = \overline{\cW^-_{\epsilon}}$.
Note that $\cW^+_{\epsilon} = \ww$ when $\epsilon = 0$ ($\ww$ is already a closed convex set).\\
Let $R < \min_{W \in \cW^+_{\epsilon}} I(X;Y)$ and let $\delta>0$ be small enough such that 
$R+\delta < \min_{W \in \cW^+_{\epsilon}} I(X;Y)$. \\ 
\\
\textbf{Case (A):} 
$P_S \in \cP_0$\\ 
W.h.p., $\vecs_a \in \tau_{S}^{\eta}$ by Lemma~\ref{lem:type} for sufficiently large $n$ 
- i.e., $||P_{\vecs_a} - P_S || \leq \eta$ w.h.p. Set the value of $\eta < \epsilon$. 
Thus, it is equivalent to communication over the expanded CAVC $\cW^+_{\epsilon}$ (i.e.,
closure of both families of channels for the CAVC is same and equal to $\cW^+_{\epsilon}$) 
so we get arbitrarily small error in message decoding. 
In particular, let $\epsilon = 3 \eta$.\\
\\
\textbf{Case (B):}
$P_S \notin \cP_0$\\
We further divide this case into two sub-cases: \\
i) $P_{\vecs_a} \in \cP_{\epsilon}^+$: Similar to Case (A), message decoding is correct and successful w.h.p. \\
ii) $P_{\vecs_a} \notin \cP_{\epsilon}^+$: Note that since $\vecs_a \in \tau_S^{\eta}$ whp and 
$P_{\vecs_a} \notin \cP_{\epsilon}^+$, we can see that that $P_{\vecs_b} \notin \cP_0$ whp. 
In fact, the following is also true 
$$ \forall P \in \cP_1 \cup \cP_2, \ ||P - P_{\vecs_b}|| \leq \frac{\eta}{2} \  \implies \ P \notin \cP_0. $$
Also, note that \eqref{eq:part-2} still remains valid even if $W_{P_{\vecs_b}} \notin \cW^+_{\epsilon}$. In other words,
for any attack vector $\vecs_b$, we still have $\eqref{eq:part-2}$ as it is a very low probability event that
a codeword which wasn't transmitted has high mutual information with the received vector $\hat \vecy$.  
Hence, w.h.p. the message decoder  would not output a wrong message - it may either decode correctly or declare $\bot$. 
If the decoder outputs $\bot$, then we identify the adversary by $G(\tilde \vecy)$- since $P_{\vecs_b} \notin \cP_0$, 
Lemma~\ref{lb-2} holds so the proof of achievability of Lemma~\ref{lem:ad-id-ach} holds as well.

Since $\epsilon$ can be made arbitrarily small, the lemma follows.
\end{proof}
\vspace{20pt}

%% file: det_ach.tex
\subsection{Achievability Proofs Under Deterministic Coding} \label{subs:det-ach}

Let, for channels $W:\cX\times\cS \to \cY$,
\begin{equation}
    \mathcal{C}_{\eta} = \{  P_{XSY}: D(P_{XSY} || P_X \times P_S \times W) \leq \eta , P_S \in \cP_1 \cup \cP_2 \},
\end{equation}
and let $$I(P) = \min_{W \in \wu, P_X = P} I(X;Y).$$

The following two lemmas establish the fact that the capacity expressions 
are indeed positive when the claimed necessary conditions are met. 
\begin{lemma}
If the channel is non-any-symmetrizable, then $\min_{W \in \wu} I(X;Y) > 0$ 
for all $P_X$ such that $P_X(x) > 0 \ \forall x \in \cX$.
\end{lemma} 
\begin{proof}
Suppose the statement is false, then there exists $P_X$ and $P_S \in \cP_1 \cup \cP_2$ for which $I(X;Y)=0$.
Hence, there exists distribution $P_{XSY} \in \cC_0$ such that $X$ and $Y$ are independent, 
i.e., $P_{Y|X}(y|x) = \sum_{s} W(y|x,s)P_S(s) = P_Y(y) \ \forall x,y$.
The C-AVC is cis-symmetrizable in a trivial manner using $U(.|x) = P_S(.)$ in \eqref{eq:cissym-csiszar}, 
a contradiction.
\end{proof}

\begin{lemma}
If the channel is non-trans-symmetrizable, then $\min_{W \in \ww} I(X;Y) > 0$ 
for all $P_X$ such that $P_X(x) > 0 \ \forall x \in \cX$.
\end{lemma} 
\begin{proof}
If $\ww = \emptyset$ then the lemma is trivially true.
If $\ww \neq \emptyset$ then,
let $\cP_0 = \{P \in \cP_1 \cup \cP_2: W_{P} \in \ww \}$. 
Suppose the statement is false, then there exists 
$P_X$ and $P_S \in \cP_0$ for which $I(X;Y)=0$.
Hence, there exists distribution $P_{XSY} \in \cC_0$ such that $X$ and $Y$ are independent, 
i.e., $P_{Y|X}(y|x) = \sum_{s} W(y|x,s)P_S(s) = P_Y(y) \ \forall x,y$.
If $P_S \in \cP_k$, then there exists $P_{S'} \in \cP_{3-k}$ such that
$W_{P_S} = W_{P_{S'}}$ as $\ww \neq \emptyset$. 
The C-AVC is trans-symmetrizable in a trivial manner using $U(.|x) = P_S(.)$ and 
$V(.|x) = P_{S'}(.)$ in \eqref{eq:transsym}, a contradiction.
\end{proof}
\vspace{20pt}

For the achievability arguments, we describe some lemmas below. 
We first present a lemma based on \cite[Lemma 3]{avc-csiszar-narayan}.

\begin{lemma} \label{lem:nice-code}
For any $\epsilon > 0$, $n \geq n_0(\epsilon)$, $N\geq \exp(n\epsilon)$, 
and type $P$, there exists codewords $\vecx_1,\vecx_2,..,\vecx_N$ in $\mathcal{X}^n$, 
each of type $P$, such that for every $\vecx \in \mathcal{X}^n$, $\vecs \in \mathcal{S}_1^n \cup \mathcal{S}_2^n$, 
and every joint type $P_{XX'S}$ (with $P_S \in \cP_1 \cup \cP_2$), upon setting $R = \frac{1}{n} \log N$, we have:
\begin{align}\label{eq:nice-code-1}
\left| \left\{ j:(\vecx,\vecx_j,\vecs)\in \tau_{XX'S}  \right\} \right|  &\leq \exp \left\{ n \left( |R - I(X';XS)|^+ + \epsilon \right) \right\}; \\
\label{eq:nice-code-2}
\frac{1}{N}\left| \left\{ i:(\vecx_i,\vecs)\in \tau_{XS}  \right\} \right| &\leq \exp (-n\epsilon/2 ) \text{, if } I(X;S)>\epsilon; \\
\label{eq:nice-code-3}
\frac{1}{N}\left| \left\{ i:(\vecx_i,\vecx_j,\vecs)\in \tau_{XX'S} \text{ for some $j \neq i$ } \right\} \right| &\leq \exp (-n\epsilon/2) \text{, if }  I(X;X'S) - |R - I(X';S)|^+ > \epsilon.
\end{align}
\end{lemma}
\begin{proof}
One can directly use \cite[Lemma 3]{avc-csiszar-narayan} to get the above result for a wider class of attacks by letting $\vecs \in (\cS_1 \cup \cS_2)^n$.
\end{proof}

\vspace{20pt}

\begin{lemma} \label{lem:cdboth-ach}
If the \cavc\ is non-any-symmetrizable and $\ww = \emptyset$ then 
$$ \Cdetboth \geq \max_{P_X} \min_{W \in \wu} I(X;Y).$$
\end{lemma}
\begin{proof}
The decoder we use for achieving the capacity is described below for $\eta$ described later. 

\begin{dec*} 
Given codewords $\vecx_j$, $j=1,\ldots,M$, set $\pand(\vecy) = (i,\sigma_k)$, $i \in \cM, k\in\{1,2\}$, iff an $\vecs \in \cS_k^n$ exists such that:
\begin{enumerate}
    \item the joint type $P_{\vecx_i,\vecs,\vecy} \in \mathcal{C}_{\eta}$, and
    \item for each  $\vecx_j, j \neq i$ such that there exists $\vecs' \in \cS_{1}^n \cup \cS_2^n$, 
    $P_{\vecx_j,\vecs',\vecy} \in \mathcal{C}_{\eta}$, we have $I(XY;X'|S) \leq \eta$ where $P_{XX'SY} = P_{\vecx_i,\vecx_j,\vecs,\vecy}$.
\end{enumerate}
Set $\pand(\vecy)=(1,\sigma_1)$ if no such $(i,\sigma_k)$ exists.
\end{dec*}

First, we justify the consistency of the decoder 
- if $(i,\sigma_k)$ satisfies both the conditions then $(i',\sigma_{k'}), (i',k') \neq (i,k)$ 
can not satisfy the conditions. Consider the following three cases
\begin{enumerate}
    \item $i\neq i',\ k \neq k'$, or
    \item $i\neq i',\ k=k'$, or
    \item $i=i',\ k\neq k'$.
\end{enumerate}

Based on \cite[Lemma 4]{avc-csiszar-narayan}, we state the following two lemmas (proved later).

\begin{lemma} \label{lem:disamb-id}
If the \cavc\ is non-trans-symmetrizable and $\beta>0$, then for a sufficiently small $\eta$, 
no  quintuple of random variables $X,X',S,S',Y$, with $P_S \in \cP_1$ and $P_{S'} \in \cP_2$, can simultaneously satisfy 
\begin{align}
    P_X = P_{X'}=P &\text{ with } \min_{a \in \mathcal{X}} P(a) \geq \beta \\
    P_{XSY} \in \mathcal{C}_{\eta}, &\text{ } P_{X'S'Y} \in \mathcal{C}_{\eta} \\
    I(XY;X'|S)\leq \eta , &\text{ } I(X'Y;X|S')\leq \eta.
\end{align}
\end{lemma}

\begin{lemma} \label{lem:disamb-comm}
    If the \cavc\ is non-any-symmetrizable and $\beta>0$, then for a sufficiently small $\eta$, 
    no  quintuple of random variables $X,X',S,S',Y$, with $P_S,P_{S'} \in \cP_1 \cup \cP_2$, can simultaneously satisfy 
    \begin{align}
        P_X = P_{X'}=P &\text{ with } \min_{a \in \mathcal{X}} P(a) \geq \beta \\
        P_{XSY} \in \mathcal{C}_{\eta}, &\text{ } P_{X'S'Y} \in \mathcal{C}_{\eta} \\
        I(XY;X'|S)\leq \eta , &\text{ } I(X'Y;X|S')\leq \eta.
    \end{align} 
\end{lemma}

Case (1) can not occur as by Lemma~\ref{lem:disamb-comm} 
(Lemma~\ref{lem:disamb-id} can also be used), as it is impossible that first and 
second condition of decoder holds for both tuples $(i,k)$ and $(i',k')$. 

Case (2) can not occur because of the same reason mentioned above.

Case (3) can not occur due to $\ww = \emptyset$. 
If case (3) was true then $(\vecx,\vecs,\vecy) \in \cC_{\eta}$ and $(\vecx,\vecs',\vecy) \in \cC_{\eta}$. 
Let $X,S,S',Y$ be random variables defined by $(\vecx,\vecs,\vecs',\vecy) \in \tau_{XSS'Y} $. 
Using Pinkser's inequality, the definition of $\mathcal{C}_{\eta}$ and the fact that 
divergence won't increase if we project $P_{XSY}$ and $P_X \times P_S \times W$ on $\mathcal{X} \times \mathcal{Y}$,
\begin{align}
    \sum_{a,c} |P_{XY}(a,c) - \sum_b P_X(a)P_{S}(b)W(c|a,b)| 
    &\leq c\sqrt{ \eta } \\
    \sum_{a,c} |P_{XY}(a,c) - \sum_b P_X(a)P_{S'}(b)W(c|a,b)| 
    &\leq c\sqrt{ \eta } \\
    \sum_{a,c} |P_X(a)U(c|a) - P_X(a)V(c|a)| 
    &\leq 2c\sqrt{ \eta } ,
\end{align}
where $U(c|a) := \sum_b P_S(b) W(c|a,b) \in \mathcal{W}_1$ and similarly $V(c|a) \in \mathcal{W}_2$. 
If $\min_a P_X(a) = \beta$ then
\begin{align} \label{contra}
\max_{a,c} |U(c|a) - V(c|a)| &\leq \frac{2c\sqrt{ \eta }}{ \beta}  .
\end{align}
However, we know that $\wa$ and $\wb$ are disjoint so \eqref{contra} is not possible by setting $\eta$ 
to be small enough and hence, a contradiction. Choose $\eta$ sufficiently small so that \eqref{contra} is 
not true and Lemma~\ref{lem:disamb-id} and \ref{lem:disamb-comm} are satisfied.

We need to show that the correct output indeed satisfies the decoding conditions with high probability. 
For this, we can show that the actual input sequence $\vecx$ and the \avcstate\ sequence $\vecs$ which was present 
in the transmission does indeed satisfy the decoder criteria. 
We prove this based on \cite[Lemma 5]{avc-csiszar-narayan}.

For any arbitrarily small $\delta>0$, choose $R$ satisfying 
\begin{equation} \label{eq:assume}
    I(P) - \delta < R < I(P)-\frac{2}{3}\delta.
\end{equation}
Choose the codebook based on Lemma~\ref{lem:nice-code} with rate $R$ and codewords $\vecx_1,\ldots,\vecx_M$.
We analyze the error probability when 
the \avcstate\ sequence is $\vecs \in \cS_t^n$ and the \cs\ is $\sigma_t$, $t=1,2$. 
Since $\ww = \emptyset$, we can define the probability of error under \avcstate\ sequence $\vecs$ as shown below
\begin{align}
    P_{e}^d(f,\phi,\vecs) &= \frac{1}{M}\sum_{i=1}^M W^n(\phi^{-1}(\{i,\sigma_t\})^C|\vecx_i,\vecs) \\
    &= \frac{1}{M}\sum_{i=1}^M \ \sum_{\vecy:\phi(\vecy) \neq (i,\sigma_t)} W^n(\vecy|\vecx_i,\vecs).
\end{align}
By \eqref{eq:nice-code-2}, 
\begin{align}
    \frac{1}{M} \lvert \{  i:(\vecx_i,\vecs) \in \bigcup\limits_{I(X;S) > \epsilon} \tau_{XS} \} \rvert &\leq (\text{no. of joint types}).\exp(-n\epsilon/2) \\
    &\leq \exp(-n\epsilon/3),
\end{align}
for suitably large $n$, which depends on the choice of $\epsilon$ which is specified later.
Therefore, it suffices to only consider codewords $\vecx_i$ for which $(\vecx_i,\vecs) \in \tau_{XS}$ with $I(X;S) \leq \epsilon$. 
If $P_{XSY} \notin \cC_{\eta}$ then,
\begin{align}
    D(P_{XSY}||P_{XS}\times W) &= D(P_{XSY}||P_X \times P_S \times W) - I(X;S) \\
    &> \eta - \epsilon .
\end{align}
Thus, 
\begin{align}
    \sum_{\vecy \in \tau_{Y|XS}(\vecx_i,\vecs)} W^n(\vecy|\vecx_i,\vecs) &\leq \exp(-nD(P_{XSY}||P_{XS}\times W) \nonumber \\
    &< \exp(-n(\eta - \epsilon)) \nonumber. \\
    \therefore \frac{1}{M} \sum_{i=1}^M \sum_{\vecy: P_{\vecx_i,\vecs,\vecy} \notin \cC_{\eta}} W^n(\vecy|\vecx_i,\vecs) &\leq \exp(-n(\eta - 2\epsilon)) \label{eq:ref1}
\end{align}

Now, if $P_{\vecx_i,\vecs,\vecy} \in \cC_{\eta} $ and yet $\phi(\vecy) \neq (i,\sigma_t)$, 
then condition (2) of the decoder must be getting violated. 
Let $\mathcal{D}_{\eta}$ be the set of all joint distributions $P_{XX'SY}$ such that 
1) $P_{XSY} \in \cC_{\eta}$; 
2) $P_{X'S'Y} \in \cC_{\eta}$; 
3) $I(XY;X'|S) > \eta$ (and $\vecx \neq \vecx'$). Then,

\begin{equation}
    \sum_{\substack{\vecy: P_{\vecx_i,\vecs,\vecy} \in \cC_{\eta}; \\ \phi(\vecy) \neq (i,\sigma_t)}} 
    W^n(\vecy|\vecx_i,\vecs) \leq \sum_{P_{XX'SY} \in \cD_{\eta}} e_{XX'SY}(i,\vecs) 
\end{equation}
where
\begin{equation} \label{def-ex}
    e_{XX'SY}(i,\vecs) = \sum_{\substack{\vecy:(\vecx_i,\vecx_j,\vecs,\vecy) \in \tau_{XX'SY}\\ \text{for some $j\neq i$}}} W^n(\vecy|\vecx_i,\vecs).
\end{equation}
Combining the equations so far, we have
\begin{equation}
    P_{e}^d(f,\phi,\vecs) \leq \exp(-n\epsilon/3) + \exp(-n(\eta-2\epsilon)) + \frac{1}{M}\sum_{i=1}^M \sum_{P_{XX'SY} \in \cD_{\eta}} e_{XX'SY}(i,\vecs).
\end{equation}

Notice that because of \eqref{eq:nice-code-3} it suffices to deal with cases when $P_{XX'SY} \in \cD_{\eta}$ satisfies
\begin{equation} \label{condn}
    I(X;X'S) \leq |R-I(X';S)|^+ + \epsilon.
\end{equation}
From \eqref{def-ex}, 
\begin{align}
    e_{XX'SY}(i,\vecs) \leq \sum_{j:(\vecx_i,\vecx_j,\vecs)\in \tau_{XX'S}} \ \sum_{\vecy \in \tau_{Y|XX'S}(\vecx_i,\vecx_j,\vecs)}W^n(\vecy|\vecx_i,\vecs) .
\end{align}
Using the fact that $W^n(\vecy|\vecx_i,\vecs)$ is a constant upper bounded by $(|\tau_{Y|XS}(\vecx_i,\vecs) |)^{-1}$, 
the inner sum is upper bounded by $|\tau_{Y|XX'S}(\vecx_i,\vecx_j, \vecs)|/ |\tau_{Y|XS}(\vecx_i,\vecs)| \leq \exp\{-n(I(Y;X'|XS)-\epsilon)\}$. 
Hence, using  \eqref{eq:nice-code-1},
\begin{equation} \label{cases}
    e_{XX'SY}(i,\vecs) \leq \exp\{ -n\left( I(Y;X'|XS) - |R-I(X';XS)|^+ -2\epsilon \right) \}.
\end{equation}
We can split the problem into two cases:
\begin{enumerate}
    \item $R \leq I(X';S)$, or,
    \item $R > I(X';S)$.
\end{enumerate}
Case (1) and \eqref{condn} yields 
\begin{align}
    I(X;X'|S) \leq I(X;X'S) \leq \epsilon,
\end{align}
and by condition (3) in definition of $\cD_{\eta}$,
\begin{align}
    I(Y;X'|XS) \geq \eta - \epsilon.
\end{align}
Since $R\leq I(X';S) \leq I(X';XS)$, it follows from \eqref{cases} that
\begin{equation}
    e_{XX'SY}(i,\vecs) \leq \exp(-n(\eta - 3\epsilon)).
\end{equation}

For case (2), from \eqref{condn}, we get
\begin{align}
    R &> I(X;X'S) + I(X';S) -\epsilon \nonumber \\
    &= I(X';XS) + I(X;S) - \epsilon \nonumber \\
    &\geq I(X';XS) - \epsilon,
\end{align}
and hence,
$$ |R - I(X';XS)|^+ \geq R - I(X';XS) -\epsilon. $$
Substituting in (\ref{cases})
\begin{align}
   e_{XX'SY}(i,\vecs) &\leq \exp\{ -n(I(X';XSY) - R - 3\epsilon)\} \label{eq:tweak}\\
   %&\leq exp\{ -n(I(X';XY|S) - R - 3\epsilon)\} \\
   %&\leq exp\{ -n(\eta - R - 3\epsilon)\}
   &\leq \exp\{ -n(I(X';Y) - R - 3\epsilon)\}.
\end{align}

$P_{XX'SY} \in \cD_{\eta}$ implies that $P_{X'S'Y} \in \cC_{\eta}$ for some $S'$. 
Thus, by definition of $\cC_{\eta}$, $P_{X'S'Y}$ is arbitrarily close to 
$P_{X''S''Y} \in \cC_0$ defined by $P_{X''S''Y''} = P_X \times P_{S'} \times W$ if $\eta$ is sufficiently small. 
This implies $I(X';Y)$ is arbitrarily close to $I(X'';Y'')$, i.e., $I(X':Y) \geq I(X'';Y'') - \delta/3$. 
By definition of $I(P)$ and assumption \eqref{eq:assume}, 
$$ I(X';Y) - R \geq I(P) - \delta/3 - R \geq \delta/3 $$
if $\eta$ is sufficiently small and depends only on $\delta$ (and $\wa,\wb$). Therefore, for case (2),
$$ e_{XX'SY}(i,\vecs) \leq \exp \{ -n (\frac{\delta}{3} - 3\epsilon)\}$$
Therefore,
$$ P_{e}^d(f,\phi,\vecs) \leq \exp(-n\epsilon/4)$$
if $\epsilon \leq \min (\eta/4,\delta/10)$ and $n$ sufficiently large for all $\vecs$. 
\end{proof}
\vspace{20pt}

\textit{Proof of Lemma~\ref{lem:disamb-id}} : Suppose there exists $X,X',S,S',Y$ which simultaneously satisfy 
the three conditions. Then, by definition of $\cC_{\eta}$,
$$ D(P_{XSY}||P_X \times P_S \times W) = \sum_{x,s,y} P_{XSY}(x,s,y) \log \frac{P_{XSY}(x,s,y)}{P_X(x)P_S(s)W(y|x,s)} \leq \eta. $$
Adding $I(XY;X'|S)$ to it,
$$D(P_{XX'SY} || P_X \times P_{X'} \times P_{S|X'} \times W) \leq 2\eta .$$
Projecting both the distributions to $\cX \times \cX \times \cY$, the divergence can not increase,
$$D(P_{XX'Y} || P_X \times P_{X'} \times V) \leq 2\eta $$
where $V(y|x,x') = \sum_s W(y|x,s)P_{S|X'}(s|x')$. By Pinsker's inequality,
\begin{equation} \label{eq:pins-1}
    \sum_{x,x',y} |P_{XX'Y}(x,x',y) - P(x)P(x')V(y|x,x')| \leq c \sqrt{2\eta}.
\end{equation}
Similarly, starting with $P_{X'S'Y} \in \cC_{\eta}$ and $I(X'Y;X|S') \leq  \eta$, we get 
\begin{equation} \label{eq:pins-2}
    \sum_{x,x',y} |P_{XX'Y}(x,x',y) - P(x)P(x')V'(y|x,x')| \leq c \sqrt{2\eta}
\end{equation}
where $V'(y|x,x') = \sum_s W(y|x',s)P_{S'|X}(s|x)$. From \eqref{eq:pins-1} and \eqref{eq:pins-2},
\begin{equation} \label{eq:contra}
    \max_{x,x',y} |V(y|x,x') - V'(y|x,x')| \leq \frac{2c \sqrt{2\eta}}{\beta^2}.
\end{equation}
For a non-trans-symmetrizable \cavc , there exists a $\xi$ such that 
\begin{equation} 
    \max_{x,x',y} |\sum_s W(y|x,s)U_{S|X}(s|x') - \sum_s W(y|x',s)V_{S|X}(s|x)| \geq \xi
\end{equation}
for every $U_{S|X} \in \cP_{A|X},\ V_{S|X} \in \cP_{B|X}$. 
Setting $U_{S|X'} = P_{S|X},V_{S|X} = P_{S'|X}$ and $\eta < \frac{\xi^2\beta^4}{8c^2}$, we get a contradiction.
Lemma~\ref{lem:disamb-comm} can be proved in a similar manner as Lemma~\ref{lem:disamb-id}. \hspace{290pt} $\blacksquare$

\vspace{20pt}

\begin{lemma}
If the \cavc\ is non-any-symmetrizable then 
$$ \Cdetcomm \geq \max_{P_X} \min_{W \in \wu} I(X;Y).$$
\end{lemma}
\begin{proof}
    The proof is analogous to the proof of Lemma~\ref{lem:cdboth-ach}. 
    We use Lemma~\ref{lem:nice-code} to get a codebook with type $P_X$ which
    maximizes $I(P)$ and use the following decoder to obtain the message.
    \begin{dec*}
    Given codewords $\vecx_j$, $j=1,\ldots,M$, set $\phi(\vecy) = i$, $i \in \cM$, iff an $\vecs$ exists such that:
    \begin{enumerate}
        \item the joint type $P_{\vecx_i,\vecs,\vecy} \in \mathcal{C}_{\eta}$, and
        \item for each  $\vecx_j, j \neq i$ such that there exists $\vecs'$, 
        $P_{\vecx_j,\vecs',\vecy} \in \mathcal{C}_{\eta}$, we have $I(XY;X'|S) \leq \eta$ 
        where $P_{XX'SY} = P_{\vecx_i,\vecx_j,\vecs,\vecy}$.
    \end{enumerate}
    Set $\phi(\vecy)=1$ if no such $i$ exists.
    \end{dec*}
\end{proof}
\vspace{20pt}

%% file: det_auth.tex
Next, we show that a positive rate is attainable for `communication or \cs\ identification'
if the \cavc\ is non-trans-symmetrizable.
\begin{lemma} \label{lem:ach-auth}
    If \cavc\ is non-trans-symmetrizable then $\Cdetauth > 0$.
\end{lemma}
\begin{proof}
Use Lemma~\ref{lem:nice-code} to obtain a codebook at some rate $R>0$ (described later). 
\begin{dec*} \label{dec:auth-d}
Given codewords $\vecx_j$, $j=1,\ldots,M$, let $B_k$ ($k=1,2$) be the set of messages $m \in \cM$ such that 
\begin{enumerate}
    \item $\exists \vecs \in \cS_k^n$ such that $P_{\vecx_m,\vecs,\vecy} \in \mathcal{C}_{\eta}$, and
    \item for every $m'\neq m$ such that $\exists\  \vecs' \in \cS_{3-k}^n$, $P_{\vecx_{m'},\vecs',\vecy} \in \mathcal{C}_{\eta}$, we have $I(XY;X'|S) \leq \eta$ where $P_{XX'SY} = P_{\vecx_m,\vecx_{m'},\vecs,\vecy}$.
\end{enumerate}
If $B_1 = B_2 = \{m\}$, then $\por(\vecy)=m$. If for some $k\in\{1,2\}$,
$B_k=\emptyset \neq B_{3-k}$, then 
the decoder outputs the compound state $\por(\vecy)=\sigma_{3-k}$.
\end{dec*}
By Lemma~\ref{lem:disamb-id}, it is not possible to have distinct messages in the sets $B_1$ and $B_2$.
Thus, the only four possibilities are listed below
\begin{enumerate}
    \item $B_1 = B_2 = \{m\}$, $m \in \cM$,
    \item $B_1 = \emptyset$, $|B_2| \geq 1$, 
    \item $B_2 = \emptyset$, $|B_1| \geq 1$, and
    \item $B_1 = B_2 = \emptyset$.
\end{enumerate}  
Suppose the \avcstate\ sequence during
the transmission is $\vecs \in \cS_t^n,\ t\in\{1,2\}$. Using the same approach as that of the proof of Lemma~\ref{lem:cdboth-ach},
we can show that the correct message would be present in the set $B_t$ w.h.p. for sufficiently large block length.
To see this, refer to the proof of Lemma~\ref{lem:cdboth-ach} - proof till \eqref{eq:ref1} remains the same.
The slightly different decoder changes the error event slightly and we present the new condition below.

If $P_{\vecx_i,\vecs,\vecy} \in \cC_{\eta} $ and yet 
$\phi(\vecy) \neq i$, then condition (2) of the decoder must be getting violated. 
Let $\mathcal{D}_{\eta}'$ be the set of all joint distributions $P_{XX'SY}$ such that 
1) $P_{XSY} \in \cC_{\eta}$; 
2) $P_{X'S'Y} \in \cC_{\eta},\ P_{S'} \in \cP_{3-t}$; 
3) $I(XY;X'|S) > \eta$ (and $\vecx \neq \vecx'$).
With this modified $\cD_{\eta}'$ definition, the rest of the proof remains the same 
till equation \eqref{eq:tweak} where we make a slight modification as shown below,
\begin{align}
    e_{XX'SY}(i,\vecs) &\leq \exp\{ -n(I(X';XSY) - R - 3\epsilon)\}\\
    &\leq \exp\{ -n(I(X';XY|S) - R - 3\epsilon)\} \\
    &\leq \exp\{ -n(\eta - R - 3\epsilon)\}, \label{eq:tough}
\end{align}
where \eqref{eq:tough} follows from definition of $\cD_{\eta}'$. Choose $0 < R = \epsilon < \eta/5$. Therefore,
$\Cdetauth >0 $. 
\end{proof}

\begin{lemma}
    If \cavc\ is non-trans-symmetrizable then
    $$ \Cdetauth \geq \max_{P_X} \min_{W \in \ww} I(X;Y). $$
\end{lemma}
\begin{proof}
For some achievable rate $R$ and block-length $n$ under random coding, apply \cite[Lemma 12.8]{csiszar-korner-book} 
to show the existence of a random code distributed over $K=n^2$ encoder-decoder pairs uniformly. 
This small amount of shared randomness can be established using deterministic codes given by Lemma~\ref{lem:ach-auth}.
Thus, we can show that $\Cdetauth = \Cranauth$ when the \cavc\ is non-trans-symmetrizable.  
\end{proof}

\subsection{Converses for Deterministic Coding}
The converses of random coding results in Section~\ref{subs:ran-conv} establish some of the
coverse results for deterministic coding.

\begin{lemma} \label{lem:converses}
    If \cavc\ is any-symmetrizable or $\ww \neq \emptyset$ then $\Cdetboth = 0 $.
\end{lemma}
\begin{proof}
Let the codewords be $\vecx_1,..,\vecx_M$. For any distribution $R(\vecs)$ over $\cS_1^n$,  
\begin{equation} \label{eq:enc-ineq}
    P_{e}^d(f,\phi,1) \geq \sum_{\vecs} R(\vecs)P_{e}^d(f,\phi,\vecs).
\end{equation} 
Let $T^n(\vecs|\vecx) = \prod_i T(s_i|x_i)$ be some distribution specified later. Choose 
\begin{equation} \label{eq:R-choice}
    R(\vecs) = \frac{1}{M}\sum_{i=1}^M T^n(\vecs|\vecx_i).  
\end{equation}
Then combining definition of $P_{e}^d(f,\phi,1)$, \eqref{eq:enc-ineq}, and \eqref{eq:R-choice},

\begin{align}
P_{e}^d(f,\phi,1) &\geq \sum_{\vecs} \left( \frac{1}{M}\sum_{i=1}^M T^n(\vecs|\vecx_i) \right) 
        \left( \frac{1}{M} \sum_{j=1}^M W^n(\phi^{-1}((j,\sigma_1))^C|\vecx_j,\vecs) \right) \\
\label{eq:pre-cases} &= \frac{1}{M^2} \sum_{i=1}^M \sum_{j =1}^M \sum_{\vecs} T^n(\vecs|\vecx_i)  W^n(\phi^{-1}((j,\sigma_1))^C|\vecx_j,\vecs) \\
\label{eq:cases} &\geq \frac{1}{M^2} \sum_{i=1}^M \sum_{j \neq i} \sum_{\vecs} T^n(\vecs|\vecx_i)  W^n(\phi^{-1}((j,\sigma_1))^C|\vecx_j,\vecs).  
\end{align}

We can have 3 cases:
\begin{enumerate}[label=(\Alph*)]
    \item the \cavc\ is trans-symmetrizable, or,
    \item the \cavc\ is cis-symmetrizable, or,
    \item $\cW_0 \neq \emptyset$.
\end{enumerate}
For case (A), let $U(s|x)$ and $V(s|x)$ be the distributions satisfying trans-symmetrizibility condition. 
Let $T(s|x) = U(s|x)$. By trans-symmetrizability condition on \eqref{eq:cases},
\begin{align}
    \frac{1}{M^2} \sum_{i=1}^M \sum_{j \neq i} \sum_{\vecs} U^n(\vecs|\vecx_i)  W^n(\phi^{-1}((j,\sigma_1))^C|\vecx_j,\vecs) 
        &=  \frac{1}{M^2} \sum_{i=1}^M \sum_{j \neq i} \sum_{\vecs} V^n(\vecs|\vecx_j)  W^n(\phi^{-1}((j,\sigma_1))^C|\vecx_i,\vecs) \\
    &\geq  \frac{1}{M^2} \sum_{i=1}^M \sum_{j \neq i} \sum_{\vecs} V^n(\vecs|\vecx_j)  W^n(\phi^{-1}((i,\sigma_2))|\vecx_i,\vecs) \nonumber \\
    &=  \frac{M-1}{M} - \frac{1}{M^2} \sum_{i=1}^M \sum_{j \neq i} \sum_{\vecs} V^n(\vecs|\vecx_j)  W^n(\phi^{-1}((i,\sigma_2))^C|\vecx_i,\vecs) \\
    &\geq  \frac{M-1}{M} - \frac{1}{M^2} \sum_{i=1}^M \sum_{j =1}^M \sum_{\vecs} V^n(\vecs|\vecx_j)  W^n(\phi^{-1}((i,\sigma_2))^C|\vecx_i,\vecs) \\
    \text{(note that $V^n(\vecs|\vecx)$ is non-zero only over $\vecs \in \cS_2^n$)} \\
    &=  \frac{M-1}{M} - P_{e}^d(f,\phi,2). \\
    \therefore \ \  P_{e}^d(f,\phi,1) + P_{e}^d(f,\phi,2) &\geq \frac{M-1}{M}. \\
    \implies P_{e}^d(f,\phi) &\geq \frac{M-1}{2M}.
\end{align}
    
Similarly, for case (B), let $U(s|x)$ and $V(s|x)$ be the distributions satisfying cis-symmetrizibility condition 
(without loss of generality we assume $\sigma_1$-symmetrizable). Let $T(s|x) = U(s|x)$. 
By performing similar steps, one can get the following inequality 
$$ P_{e}^d(f,\phi,1) \geq \frac{M-1}{2M}.$$
$$\therefore \ \  P_{e}^d(f,\phi) \geq \frac{M-1}{2M}. $$

For case (C), say $Z_{Y|X} \in \ww$. 
Let $P_k(s)$ be a distribution over $\cS_k$ be such that $\sum_s P_k(s)W_{Y|X,S=s} = Z_{Y|X}$. 
Set $T(s|x) = P_1(s)$. Simplifying \eqref{eq:pre-cases}, we get
$$ P_{e}^d(f,\phi,1) \geq \frac{1}{M} \sum_{i=1}^M  Z^n(\phi^{-1}((i,\sigma_1))^C|\vecx_i). $$
Similarly, setting $T(s|x) = P_2(s)$, we get,
$$ P_{e}^d(f,\phi,2) \geq \frac{1}{M} \sum_{i=1}^M  Z^n(\phi^{-1}((i,\sigma_2))^C|\vecx_i). $$
Adding both,
$$P_{e}^d(f,\phi,1) + P_{es}(f,\phi,2) \geq 1. $$
$$\therefore \ \ P_{e}^d(f,\phi) \geq \frac{1}{2}. $$

Therefore, non-any-symmetrizability and $\ww = \emptyset$ is necessary for non-zero rate of 
communication and \cs\ identification. 
\end{proof}
Similar steps can be performed to show that any-symmetrizability implies $\Cdetcomm = 0$.
\vspace{20pt}

\begin{lemma} \label{lem:cdauth-pos}
If \cavc\ is trans-symmetrizable then $\Cdetauth = 0$.
\end{lemma}
Steps  similar to proof of Lemma~\ref{lem:converses} can be used to show
that trans-symmetrizability leads to the condition $P^d_e(f,\phi) \geq \frac{M-1}{2M}$.
\vspace{10pt}

%% file: lemma-proof.tex
We now give the proof for Lemma~\ref{lb-1}.

Consider the channel $\widetilde Z_{Y|X}$ which is the weighted average of the individual 
channels $W_{Y|X,S=s}$ (weighted with respect to fraction of $s \in \cS_k$ occurrences, formalized later). 
We prove that the input $\vecx$, which is in the typical set $\tau_{X}$, and the output $\vecy$ would be jointly typical with respect to 
the distribution $P_X \times \widetilde Z_{Y|X}$.

Without loss of generality, we analyze the problem when $\vecs \in \cS_1^n$. 
Let $\cS_1 = \{S_1,S_2,\ldots,S_{T}\}$ (where $T = |\cS_1|$). Denote the indices of 
$\vecs \in \cS^n_1$ where $s = S_i$ as $J_{i}(\vecs)$, ie, 
$J_{i}(\vecs) = \{j: s_j = S_i\}$. Notice that,

\begin{align} 
    P(\vecy,\vecx|\vecs) &= \frac{1}{|\tau_{X}|}W^n(\vecy|\vecx,\vecs) \\
     &= \frac{1}{|\tau_{X}|}\prod_{i=1}^n W(y_i|x_i,s_i) \\
    \label{sq_eqn}    &= \frac{1}{|\tau_{X}|}\prod_{i=1}^T \left[\prod_{j \in J_{i}(\vecs)}W(y_j|x_j,S_i)\right].
\end{align}

Fix an $\epsilon_1$ (value described later) and from the sets $J_{i}(\vecs)$, consider the sets which have 
$|J_{i}(\vecs)|>\epsilon_1 n$, i.e., $\mathcal{G} := \{i \in \{1,2,\ldots,T\}:|J_{i}(\vecs)|>\epsilon_1 n \}$. 
$\mathcal{G}$ is non-empty for any value of $\epsilon_1 < 1/T$. Choose any $\epsilon_1 < \min\{1/T,1/T'\}$ where $T' = |\cS_2|$. 
Henceforth, we shall assume $\epsilon_1$ satisfies this condition.  
Define the `subset' vectors $ \vecx_i := \{x_j: j \in J_{i}(\vecs)\}$ and similarly $\vecy_i$. Let $\vecS_i$ be 
the vector of $|J_i(\vecs)|$ repetitions of symbol $S_i$. Then, we can write \eqref{sq_eqn} as
$$P(\vecy,\vecx|\vecs) = \frac{1}{|\tau_{X}|}\prod_{i=1}^T \left[W^{|J_{i}(\vecs)|}(\vecy_i|\vecx_i,\vecS_i)\right].$$
By Lemma~\ref{lem:type}, $\vecx_i,\ i \in \mathcal{G} $ are of type $\tau^{\epsilon_2}_{X}$  
with probability greater than $1-f(\epsilon_2)$ for arbitrarily small $\epsilon_2$ and sufficiently large $n$ as their lengths 
are at least $\epsilon_1n$ and $f(\cdot)$ satisfies $f(\epsilon_2) \to 0$ as $\epsilon_2 \to 0$. Therefore,
$ P\{ \vecx_i \in \tau^{\epsilon_2}_{X} \  \forall i \in \mathcal{G} \} \geq 1 - f_2(\epsilon_2)$
where $f_2(\cdot) = |\cG|f(\cdot)$ satisfies $f_2(\epsilon_2) \to 0$ as $\epsilon_2 \to 0$.

By conditional typicality lemma, if random variables $X,Y_k$ are distributed as $P_{XY_k} = P_X \times W_{Y|X,S=S_k}$, then
for any $\epsilon_4>0$
\begin{equation}
P\{(\vecx_i,\vecy_i) \in \tau^{\epsilon_3}_{XY_k} | \vecx_i \in \tau^{\epsilon_2}_{X} \} > 1 - \epsilon_4,  \forall i \in \mathcal{G}    
\end{equation}
for any $\epsilon_3 > \epsilon_2$ and sufficiently large $n$. Denote the event 
$\{(\vecx_i,\vecy_i) \in \tau^{\epsilon_3}_{XY_k} \forall i \in \mathcal{G} | \vecx_i \in \tau^{\epsilon_2}_{X} \forall i \in \mathcal{G} \} = \cB$. 
Similarly,
$$ P(\mathcal{B}) > 1-|\cG|\epsilon_3.$$
%for a function $f_3(\cdot)$ which satisfies $f_3(\epsilon_3) \to 0$ as $\epsilon_3 \to 0$.
%
Therefore, with high probability, the $(\vecx_i,\vecy_i),\ i \in \cG$ are jointly typical according to the distribution $P_X\times W_{Y|X,S=S_k}$. 
Denote $W_{Y|X,S=S_i}$ as $Z_{Y|X}^i$ (this is a single letter channel). We now show that $(\vecx,\vecy)$ is jointly typical with 
$P_X \times \widetilde Z_{Y|X}$ with high probability, where
$$\widetilde Z_{Y|X}(b|a) = \frac{1}{\sum_{i \in \mathcal{G}}|J_{i}(\vecs)|}\sum_{i \in \mathcal{G}} Z^i_{Y|X}(b|a)|J_{i}(\vecs)|,\ (a,b)\in \cX \times \cY.$$
Clearly, $ \widetilde Z_{Y|X} \in \wa$. We need to show (w.h.p.)
$$ |\pi(a,b|\vecx,\vecy) - P_X(a)\widetilde Z_{Y|X}(b)| \leq \epsilon \ \forall (a,b) \in \mathcal{X}\times \mathcal{Y}$$
where $\pi(a,b|\vecx,\vecy)$ is the empirical distribution and $\epsilon$ is specified later. 

Since $\mathcal{G}$ contains $J_{i}(\vecs)$ which have at least cardinality of $\epsilon_1 n$, 
we can say that $\sum_{i \in \mathcal{G}^C}|J_{i}(\vecs)| \leq (T-1)\epsilon_1n$. 
Therefore, $\sum_{i \in \mathcal{G}}|J_{i}(\vecs)| > n(1-(T-1)\epsilon_1)$. Hence, w.h.p.,

\begin{align}
     \pi(a,b|\vecx,\vecy)  &= \frac{1}{n} \sum_{i=1}^K|J_{i}(\vecs)|\pi(a,b|\vecx_i,\vecy_i) \\
 &= \frac{1}{n} \left( \sum_{i \in \mathcal{G}}|J_{i}(\vecs)|\pi(a,b|\vecx_i,\vecy_i) + \sum_{i \in \mathcal{G}^C}|J_{i}(\vecs)|\pi(a,b|\vecx_i,\vecy_i) \right).
\end{align}
Further,w.h.p.,
\begin{align} 
    \frac{1}{n} \left( \sum_{i \in \mathcal{G}}|J_{i}(\vecs)|(1-\epsilon_3)P_X(a)Z^i_{Y|X}(b|a) \right) 
&\leq \pi(a,b|\vecx,\vecy) \leq \frac{1}{n} \left( \sum_{i \in \mathcal{G}}|J_{i}(\vecs)|(1+\epsilon_3)P_X(a)Z^i_{Y|X}(b|a) + \sum_{i \in \mathcal{G}^C}\epsilon_1n \right) \\
 (1-\epsilon_3)(1-(T-1)\epsilon_1) P_X(a)\widetilde Z_{Y|X}(b|a) 
&\leq \pi(a,b|\vecx,\vecy) \leq  (1+\epsilon_3)P_X(a)\widetilde Z_{Y|X}(b|a) + (T-1)\epsilon_1.
\end{align}
Therefore (w.h.p.),
\begin{align}
    |\pi(a,b|\vecx,\vecy) - P_X(a)\widetilde Z_{Y|X}(b|a) | \leq \max \{ &(\epsilon_1(T-1) +\epsilon_3 - \epsilon_1\epsilon_3(T-1)) P_X(a)\widetilde Z_{Y|X}(b|a) , \nonumber \\ &\epsilon_3 P_X(a)\widetilde Z_{Y|X}(b|a) + (T-1)\epsilon_1 \} \nonumber 
\end{align} 
\begin{equation}
    \hspace{-20pt} \leq \epsilon_3 + (T-1)\epsilon_1.
\end{equation}

Pick $\epsilon \geq \max\{\epsilon_3 + (T-1)\epsilon_1, \epsilon_3 + (T'-1)\epsilon_1\}$. Since $\epsilon_1,\epsilon_2,\epsilon_3$ and $\epsilon_4$ 
(with $\epsilon_2 < \epsilon_3$) can be set arbitrarily small for sufficiently large $n$, we can set $\epsilon$ to be arbitrarily small as well for large $n$.

%% file: paper.bbl
\begin{thebibliography}{10}

\bibitem{blackwell-compound}
D.~Blackwell, L.~Breiman, and A.~J. Thomasian, ``The capacity of a class of
  channels,'' {\em The Annals of Mathematical Statistics}, vol.~30, no.~4,
  pp.~1229--1241, 1959.

\bibitem{blackwell-AVC}
D.~Blackwell, L.~Breiman, and A.~J. Thomasian, ``The capacities of certain
  channel classes under random coding,'' {\em The Annals of Mathematical
  Statistics}, vol.~31, no.~3, pp.~558--567, 1960.

\bibitem{wolf}
J.{Wolfowitz }, ``Simultaneous channels,'' 1959.

\bibitem{avc-csiszar-narayan}
I.~{Csiszar} and P.~{Narayan}, ``The capacity of the arbitrarily varying
  channel revisited: positivity, constraints,'' {\em IEEE Transactions on
  Information Theory}, vol.~34, no.~2, pp.~181--193, 1988.

\bibitem{NehaISIT21}
N.~Sangwan, M.~Bakshi, B.~K. Dey, and V.~Prabhakaran, ``Communication with
  adversary identification in byzantine multiple access channels,'' {\em IEEE
  International Symposium on Information Theory}, 2021.

\bibitem{Jahn1981}
J.~Jahn, ``Coding of arbitrarily varying multiuser channels,'' {\em IEEE Trans.
  Inf. Theory}, vol.~27, pp.~212--226, 1981.

\bibitem{PeregSteinberg2017}
U.~Pereg and Y.~Steinberg, ``The arbitrarily varying broadcast channel with
  degraded message sets with causal side information at the encoder.''
\newblock arXiv:1709.04770, 2017.

\bibitem{HofBross2006}
E.~{Hof} and S.~I. {Bross}, ``On the deterministic-code capacity of the
  two-user discrete memoryless arbitrarily varying general broadcast channel
  with degraded message sets,'' {\em IEEE Transactions on Information Theory},
  vol.~52, no.~11, pp.~5023--5044, 2006.

\bibitem{kosut-kliewer}
O.~{Kosut} and J.~{Kliewer}, ``Authentication capacity of adversarial
  channels,'' in {\em 2018 IEEE Information Theory Workshop (ITW)}, pp.~1--5,
  2018.

\bibitem{BKKGYu}
A.~Beemer, O.~Kosut, J.~Kliewer, E.~Graves, and P.~Yu, ``Structured coding for
  authenticationin the presence of a malicious adversary,'' {\em IEEE
  International Symposium on Information Theory}, 2019.

\bibitem{GYuS}
E.~Graves, P.~Yu, and P.~Spasojevic, ``Keyless authentication in the presence
  of asimultaneously transmitting adversary,'' {\em 2018 IEEE Information
  Theory Workshop (ITW)}.

\bibitem{KK2}
O.~{Kosut} and J.~{Kliewer}, ``Network equivalence for a joint
  compound-arbitrarily-varying network model,'' {\em IEEE Information Theory
  Workshop (ITW), 2016}.

\bibitem{csiszar-korner-book}
I.~Csiszar and J.~Korner, {\em Information Theory: Coding Theorems for Discrete
  Memoryless Systems}.
\newblock USA: Academic Press, Inc., 1982.

\bibitem{skala2013hypergeometric}
M.~Skala, ``Hypergeometric tail inequalities: ending the insanity,'' 2013.

\end{thebibliography}
